\documentclass[11pt]{article}
\usepackage[margin=1in]{geometry}
\usepackage{amsmath,amssymb,amsthm}
\usepackage[numbers,sort&compress]{natbib}
\usepackage{graphicx}
\usepackage{booktabs}
\usepackage{enumitem}
\usepackage{microtype}
\usepackage{xcolor}
\usepackage[colorlinks=true,linkcolor=blue!50!black,citecolor=blue!50!black,urlcolor=blue!50!black]{hyperref}

\theoremstyle{plain}
\newtheorem{theorem}{Theorem}[section]
\newtheorem{proposition}[theorem]{Proposition}
\newtheorem{lemma}[theorem]{Lemma}
\newtheorem{corollary}[theorem]{Corollary}
\newtheorem{fact}[theorem]{Fact}
\theoremstyle{definition}
\newtheorem{definition}[theorem]{Definition}

\theoremstyle{remark}

\newcommand{\E}{\mathbb{E}}
\newcommand{\Var}{\operatorname{Var}}
\newcommand{\Cov}{\operatorname{Cov}}
\newcommand{\Prob}{\mathbb{P}}
\newcommand{\R}{\mathbb{R}}
\newcommand{\ind}[1]{\mathbf{1}\{#1\}}
\newcommand{\risk}{R}
\newcommand{\brisk}{r}

\newcommand{\Corr}{\mathrm{Corr}}
\newcommand{\TV}{\mathrm{TV}}
\newcommand{\MSE}{\mathrm{MSE}}
\newcommand{\cP}{\mathcal{P}}
\newcommand{\cT}{\mathcal{T}}
\newcommand{\cD}{\mathcal{D}}

\newcommand{\Lab}{[q]^n}
\newcommand{\pbar}{\bar{p}}
\newcommand{\dt}{\,\mathrm{d}}
\newcommand{\vq}{v_q}
\newcommand{\triv}{v}
\newcommand{\Lam}{\Lambda_n}
\newcommand{\Ein}{E_{\mathrm{in}}}
\newcommand{\Eout}{E_{\mathrm{out}}}
\newcommand{\nE}{|E|}
\newcommand{\Diag}{\operatorname{Diag}}
\newcommand{\tr}{\operatorname{tr}}
\newcommand{\overlap}{\operatorname{ov}}
\newcommand{\sgn}{\operatorname{sgn}}
\providecommand{\Theta}{\varTheta}
\renewcommand{\Theta}{\varTheta}

\title{\textbf{Slow Beats Fast at the Kesten--Stigum Threshold:\\ Minimax, Fisher-Information and Belief-Propagation Characterizations of the Information--Computation Gap in Sparse Stochastic Block Models}}
\author{Soroor Ghandali}
\date{October 2026}

\begin{document}
\maketitle

\begin{abstract}
We study community recovery in the sparse symmetric stochastic block model with $q$ communities, average degree $d$ and signal strength $\lambda$ through classical statistical decision theory and Fisher information, and obtain three characterizations of the Kesten--Stigum threshold $d\lambda^2=1$ and of the information--computation gap below it. First, we show that on each community-size profile the minimax risk of any class of decision rules closed under averaging and vertex relabeling equals its Bayes risk under the uniform prior, and that one random relabeling of the input gives every algorithm constant risk on each profile. The posterior mean is the unique Bayes rule, is admissible, and is the unique minimax rule on the balanced profile. The Bayes risk of the class of degree-$D$ polynomial rules is exactly $\vq(1-\Corr_D^2)$, with $\vq$ the trivial risk and $\Corr_D$ the low-degree correlation. Combined with the low-degree bound of Chin, Mossel, Sohn and Wein and the information-theoretic results of Banks, Moore, Neeman and Netrapalli and of Mossel, Sly and Sohn, this gives the gap as a worst-case statement: for $q\ge5$ there is a window below the threshold in which no low-degree rule has worst-case risk below $\vq-o(1)$, while an exponential-time rule has risk at most $\vq-\varepsilon^4/q^2$ on a set of labelings of probability $1-o(1)$. Second, the Fisher information about $\lambda$ carried by short cycle counts is a series with terms of order $k(d\lambda^2)^k$, convergent exactly when $d\lambda^2<1$; a Cram\'er--Rao bound for a restricted class of statistics shows that below the threshold the relative error of every unbiased cycle-based estimator of $\lambda^k$ stays above an explicit constant however long the cycles counted, that every cycle-count test has success probability bounded by an explicit constant below one, and that the UMVUE of $\lambda^k$ from a single cycle length attains the bound for that length; Karlin's conditions hold for the Poisson family, while on the restricted range of the model the estimator is dominated by its clipped version. Third, the derivative of belief propagation at its uninformative fixed point multiplies a random perturbation by $|\lambda|\sqrt d$ per iteration, and one EM step taken there leaves $\lambda$ unchanged, so that the parameters cannot be learned while the algorithm is stuck. A signal-to-noise computation recovers the condition $d\lambda^{1/\chi}>1$ of Chin et al.\ for $q=n^\chi$ communities and identifies personalized PageRank as a walk count with suboptimal geometric weights. Experiments on networks with up to $3\times10^5$ vertices confirm the location of the threshold at $d\lambda^2=1$ for $q=2$, the hard window for $q=5$, the planted-coloring value $c_d\approx12.84$ to within one percent, the growth factor to within $0.006$, the EM invariance, and the many-community scaling.

\end{abstract}



\section{Introduction}
\label{sec:intro}

The stochastic block model \citep{holland1983} is the canonical model of a network with hidden community structure. In its sparse symmetric form it has $n$ vertices, each carrying a hidden label in $[q]=\{1,\dots,q\}$, and each pair of vertices is joined by an edge independently, with probability $a/n$ if the two labels agree and $b/n$ otherwise. The average degree $d=(a+(q-1)b)/q$ stays bounded as $n\to\infty$, which is the regime relevant to real networks and the one in which the inference problem is hardest. The task is to recover the labels, up to a permutation of the $q$ names, from one observation of the graph. Two numbers govern the problem: the average degree $d$ and the signal strength
\begin{equation*}
\lambda=\frac{a-b}{a+(q-1)b},
\end{equation*}
the second eigenvalue of the Markov chain that transports a label across one edge. Their combination $d\lambda^2$ is the Kesten--Stigum number, after the threshold discovered by \citet{kesten1966} for multitype branching processes.

Fifteen years of work have established the following picture \citep{decelle2011,mossel2015,massoulie2014,mossel2018,bordenave2018,abbe2018,banks2016,mossel2023,abbe2017}. When $d\lambda^2>1$, polynomial-time algorithms (spectral methods based on the non-backtracking matrix, belief propagation) recover the communities with positive correlation. For $q=2$, when $d\lambda^2<1$ the graph is contiguous to an Erd\H{o}s--R\'enyi graph of the same average degree, so recovery is impossible for any method. For $q\ge5$ there is a window of parameters with $d\lambda^2<1$ in which recovery is information-theoretically possible, by an exhaustive search over labelings, but no polynomial-time algorithm is known to succeed, and lower bounds against low-degree polynomial algorithms \citep{hopkins2017,chin2025} show that a large natural class of efficient algorithms fails there. This is the information--computation gap of the block model, the prototype of such gaps in high-dimensional statistics.

The existing analyses of the gap come from probability theory and theoretical computer science: second-moment and contiguity arguments, spectral theory of random matrices, and the low-degree likelihood-ratio method. In this paper we analyze the same question with the tools of classical statistical decision theory and estimation theory: Bayes rules, minimax rules and admissibility; Fisher information, the Cram\'er--Rao bound, uniformly minimum-variance unbiased estimators (UMVUE), best linear unbiased estimators (BLUE), and Karlin's admissibility theorem. These tools turn out to describe the threshold and the gap exactly, and they yield statements of a different kind from those in the literature: worst-case rather than average-case, with explicit constants, and with a precise accounting of where the information about the parameters resides.

\subsection{Main results}
\label{sec:mainresults}

Throughout, labels are i.i.d.\ uniform on $[q]$, the quantity to be estimated for each pair of vertices is $x_{ij}=\ind{\sigma_i=\sigma_j}-1/q$, and a decision rule is scored by the mean squared error over pairs; the trivial rule $t\equiv0$ has risk $\vq=\frac1q(1-\frac1q)$. A rule is \emph{low-degree} of degree $D$ if each of its outputs is a polynomial of degree at most $D$ in the entries of the adjacency matrix; this class contains counts of small subgraphs, weighted path and walk counts, power iterations and linearized message passing, and is the standard proxy for polynomial-time computation in this problem \citep{hopkins2017,kunisky2019,schramm2022}.

\paragraph{The gap as a minimax theorem.}
By a symmetrization argument we prove that, within every class of rules that is closed under symmetrization, the minimax risk over the labelings of a given community-size profile equals the Bayes risk under the uniform prior on that profile (Theorem~\ref{thm:worst}); a one-line variant shows that \emph{every} algorithm, after one uniformly random relabeling of its input, has constant risk on each profile. The posterior mean is the unique Bayes rule and is admissible, and on the balanced profile it is the unique minimax rule, the uniform prior being least favorable (Theorem~\ref{thm:bayes}); a posterior sample, which is what an ideal Gibbs sampler returns, has exactly twice its Bayes risk (Lemma~\ref{lem:sample}). For the low-degree class we obtain the exact formula
\begin{equation*}
\inf_{\delta\in\cP_D}\ \brisk(\delta)=\vq\,\big(1-\Corr_D^2\big),
\end{equation*}
where $\Corr_D$ is the low-degree correlation introduced by \citet{schramm2022} and bounded for the block model by \citet{chin2025}, and the minimax risk of the class on each profile has the same form (Theorem~\ref{thm:lowdeg}). Combining this identity with the bound $\Corr_D\le C\sqrt{q/n}$ of Chin et al.\ and with the information-theoretic results of \citet{banks2016,abbe2018,mossel2023}, we obtain the gap as a theorem about worst-case risk (Corollary~\ref{cor:gap}): in the hard window no low-degree rule has worst-case risk below $\vq(1-C^2q/n)-o(1)$, while the information-theoretically optimal procedure, after a random relabeling, has risk at most $\vq-\varepsilon^4/q^2+o(1)$ on every labeling of a set of probability $1-o(1)$, where $1/q+\varepsilon$ is the overlap it achieves; in particular the Bayes risk of the posterior mean is at most $\vq-\varepsilon^4/q^2+o(1)$. The margin $\varepsilon^4/q^2$ is explicit (Lemma~\ref{lem:chance}).

\paragraph{The threshold from Fisher information.}
We compute the Fisher information about $\lambda$ at three levels. With the labels known it is $I_c(\lambda)=n\,d(q-1)/\big(2(1+(q-1)\lambda)(1-\lambda)\big)$, and the estimators of $a$ and $b$ are simultaneously the UMVUE, the BLUE and the maximum-likelihood estimator, attaining the Cram\'er--Rao bound exactly; moreover $d$ and $\lambda$ are orthogonal parameters (Section~\ref{sec:known}). For the observed graph, Fisher's identity expresses the score as $\kappa\sum_e(p_e-\pbar)$, a sum over edges of the excess posterior probability that the edge lies inside a community, and the information is $I_c(\lambda)(1-V)$, where $V$ is the fraction of missing information (Section~\ref{sec:unknown}). For the class of functions of short cycle counts we prove that the information is a series with terms of order $k(d\lambda^2)^k$,
\begin{equation*}
I_{\mathrm{cyc}}(\lambda)\ \le\ \frac{(q-1)^2\,d\,r^2(3-2r)}{2(1-r)^2},\qquad r=d\lambda^2<1,
\end{equation*}
and that it diverges when $r\ge1$ (Theorem~\ref{thm:cycleinfo}); these are statements about the Poisson limit experiment of the cycle counts, with constants uniform in the cycle length. A Cram\'er--Rao inequality for a restricted class of statistics (Theorem~\ref{thm:rcr}) then shows that $\lambda$ itself admits no unbiased cycle-based estimator, and that below the threshold the relative error of every unbiased cycle-based estimator of $\lambda^k$ is bounded below by a constant that does not depend on the cycle lengths used (Corollary~\ref{cor:crcyc}); a Hellinger computation caps the success probability of every cycle-count test by $\frac12(1+\sqrt{1-\rho^2})<1$ with an explicit $\rho$ (Theorem~\ref{thm:cycletest}); the UMVUE of $\lambda^k$ from the $k$-cycle count attains the Cram\'er--Rao bound and its relative error vanishes with $k$ exactly when $r>1$ (Section~\ref{sec:umvue}); and the BLUE of the cycle counts carries the full cycle information. We verify Karlin's conditions for the Poisson family, so that the cycle estimator is admissible when its mean is unrestricted, and show that on the restricted range of the block model it is inadmissible, being dominated by its clipped version (Proposition~\ref{prop:karlin}).

\paragraph{Belief propagation and empirical Bayes.}
The derivative of the belief-propagation map at its uninformative fixed point is $\lambda$ times the non-backtracking operator of the graph, corrected by the field term, which removes the uniform perturbation (the leading eigenvector, of eigenvalue $d$) that would otherwise grow by $|\lambda|d$ per iteration; on a sparse random graph a random perturbation is then multiplied by $|\lambda|\sqrt d$ per iteration (Theorem~\ref{thm:bp}). One step of the EM algorithm taken at that fixed point resets $d$ to the observed mean degree and leaves $\lambda$ unchanged for every graph (Proposition~\ref{prop:em}), so a pessimistic initial guess is self-confirming, and the number of iterations a damped algorithm spends before returning to the fixed point grows like $1/\ln(1/f)$ with $f=1-\alpha+\alpha|\lambda|\sqrt d$.

\paragraph{Many communities and PageRank.}
For $q=n^\chi$ communities a signal-to-noise computation for the walk counts between a pair of vertices gives the per-pair ratio $z_k=(d\lambda^2)^{k/2}\,q/\sqrt n$ and recovers the recovery condition $d\lambda^{1/\chi}>1$ of \citet{chin2025}, hence the rule $\lambda>d^{-\max(1/2,\chi)}$ for efficient recovery (Section~\ref{sec:many}). The same computation identifies personalized PageRank as a walk count with geometric weights $(\alpha/d)^k$ where the optimal weights are $\lambda^k$, extending the result of \citet{kloumann2017} to the sparse, unseeded, many-community setting.

\paragraph{Local learning and the curse of dimensionality.}
A lower bound for learning unstructured labels, $\mathrm{error}\ge(1-1/K)(1-n/p^D)$ for a learner that has seen $n$ of $p^D$ possible inputs, supplies the baseline against which recovery is measured and isolates what the block model adds, namely a relation between the labels of distinct inputs; the signal-to-noise computation \eqref{eq:snr} then explains why nearest-neighbor voting, $k$-means on spectral embeddings, random-walk methods and message-passing networks, all of which read bounded neighborhoods, stop at the Kesten--Stigum line (Section~\ref{sec:local}).

\paragraph{Experiments.}
All quantitative predictions are tested on simulated networks with up to $3\times10^5$ vertices (Section~\ref{sec:experiments}): the location of the threshold at $d\lambda^2=1$ for $q=2$, with a transition that steepens as $n$ grows; the hard window for $q=5$, in which the true labeling is a stable fixed point of belief propagation while every algorithm started without information stays at chance; the planted-coloring stability value $c_d\approx12.84$ of \citet{decelle2011}, reproduced to within one percent; the growth factor $|\lambda|\sqrt d$, reproduced to within $0.006$ for $d\lambda^2\le1.3$; the EM invariance, reproduced to $10^{-5}$; the exact decision-theoretic identities on a small model in which all $5{,}775$ partitions are enumerated; the scaling exponents of the critical $\lambda$ with $d$ for $q=n^\chi$, which match $\chi$ for $\chi\le0.7$; and the PageRank predictions.

\subsection{Related work}
\label{sec:related}

The Kesten--Stigum threshold originates in the theory of multitype branching processes \citep{kesten1966}; its role for reconstruction on trees was established by \citet{evans2000}. \citet{decelle2011} conjectured the full phase diagram of the sparse block model from the cavity method, introduced the belief-propagation algorithm used here, identified the hard phase for $q\ge5$, and computed the planted-coloring thresholds against which we check our simulations. The positive side of the conjecture, efficient recovery above the threshold, was proved for $q=2$ by \citet{mossel2018} and \citet{massoulie2014}, and for general $q$ by \citet{bordenave2018} and \citet{abbe2018}; the spectral methods we use are those of \citet{krzakala2013} and \citet{saade2014}. The negative side for $q=2$ is due to \citet{mossel2015}, whose contiguity argument rests on the Poisson limit of short cycle counts that we use in Section~\ref{sec:fisher}. That the threshold is not sharp for large $q$ was shown by \citet{banks2016}, who also located the information-theoretic threshold up to constants; on the tree, \citet{sly2011} showed that the Kesten--Stigum bound is the reconstruction threshold for $q=3$ at large degree and is not for any $q\ge5$ at any degree, and \citet{mossel2023} showed that the same holds for the block model, with $q=3,4$ sharp at large degree. Low-degree lower bounds for the block model are due to \citet{hopkins2017} and, in the form used here, to \citet{chin2025}; the low-degree framework for estimation is that of \citet{schramm2022}, see also \citet{kunisky2019}. The regime of many communities was opened by \citet{chin2025} and \citet{carpentier2026,carpentier2026b}. Minimax rates for community detection in the regime of growing degree were obtained by \citet{zhang2016}; our minimax statements concern the bounded-degree regime and the low-degree class. The connection between block models and personalized PageRank was made by \citet{kloumann2017} for the dense seeded two-block model. The behavior of the EM algorithm with belief propagation was studied by \citet{kawamoto2018}, and the failure of simulated annealing in planted problems by \citet{angelini2023}; semi-supervised variants were analyzed by \citet{zhang2014semi} and \citet{kanade2016}. The decision-theoretic and estimation-theoretic results we rely on, Jensen's inequality, the minimaxity of constant-risk Bayes rules, the admissibility of unique Bayes rules, the Lehmann--Scheff\'e theorem and Karlin's theorem, are classical \citep{lehmann1950,karlin1958,blyth1951,lehmann1998}; Fisher's identity for the score of incomplete data is used in the form given by \citet{louis1982} in the EM literature \citep{dempster1977}. The Poisson limit of short cycle counts is classical for Erd\H{o}s--R\'enyi graphs \citep{bollobas2001} and was extended to the block model by \citet{mossel2015} and \citet{banks2016}.

\subsection{Organization}

Section~\ref{sec:prelim} fixes the model, the notation and the definitions. Section~\ref{sec:threshold} derives the three model numbers and the threshold, and introduces a running example. Section~\ref{sec:decision} contains the decision-theoretic results, Section~\ref{sec:bp} the analysis of belief propagation and of the EM algorithm, and Section~\ref{sec:fisher} the Fisher-information results. Section~\ref{sec:many} treats many communities and PageRank, and Section~\ref{sec:local} the connection with local learning. Section~\ref{sec:experiments} reports the experiments and Section~\ref{sec:discussion} concludes. Appendix~\ref{app:notation} collects the notation and Appendix~\ref{app:algorithms} specifies the algorithms.

\section{Model, notation and definitions}
\label{sec:prelim}

\subsection{The model}

A graph on the vertex set $[n]$ is encoded by its adjacency matrix $A\in\{0,1\}^{n\times n}$, with $A_{ij}=A_{ji}=1$ if $i$ and $j$ are joined by an edge and $A_{ii}=0$. The degree of a vertex is its number of edges, and two vertices joined by an edge are neighbors. A walk of length $k$ is a sequence of $k+1$ vertices in which consecutive vertices are adjacent; it is non-backtracking if it never returns immediately to the vertex it has just left; a cycle of length $k$ is a closed walk through $k$ distinct vertices.

\begin{definition}[Sparse symmetric stochastic block model]
\label{def:sbm}
Let $n\ge q\ge2$ be integers and $a,b\ge0$ with $a\ne b$. Each vertex $i\in[n]$ receives a label $\sigma_i\in[q]$, the labels being i.i.d.\ uniform. Given $\sigma=(\sigma_1,\dots,\sigma_n)$, the pairs $\{i,j\}$ are joined independently, with probability $a/n$ if $\sigma_i=\sigma_j$ and $b/n$ if $\sigma_i\ne\sigma_j$. We write $\mathrm{SBM}(n,q,a,b)$ for the joint law of $(\sigma,A)$, and $P_\sigma$ for the conditional law of $A$ given $\sigma$. The graph $A$ is observed and $\sigma$ is not. The model is \emph{assortative} if $a>b$ and \emph{disassortative} if $a<b$; $a=0$ is the planted $q$-coloring model. The \emph{sparse regime} is $a,b,q$ fixed as $n\to\infty$.
\end{definition}

The \emph{profile} of a labeling $\sigma$ is the vector $m(\sigma)=(m_1,\dots,m_q)$ of community sizes, $m_s=\#\{i:\sigma_i=s\}$. For a profile $m$ we write $\Theta_m=\{\sigma\in\Lab: m(\sigma)=m\}$; when $q$ divides $n$, $\Theta_{\mathrm{bal}}$ denotes the orbit of the \emph{balanced} profile $m=(n/q,\dots,n/q)$. A permutation $g$ of $[n]$ acts on labelings, graphs and pair-indexed arrays by $(g\sigma)_i=\sigma_{g^{-1}(i)}$, $(gA)_{ij}=A_{g^{-1}(i)g^{-1}(j)}$ and $(gt)_{ij}=t_{g^{-1}(i)g^{-1}(j)}$. Each $\Theta_m$ is a single orbit of this action, and the law of $gA$ under $P_\sigma$ is $P_{g\sigma}$. The profile is independent of $(a,b)$ and is therefore ancillary for the edge parameters.

\begin{definition}[Model numbers]
\label{def:numbers}
\begin{equation}
d=\frac{a+(q-1)b}{q},\qquad
\lambda=\frac{a-b}{a+(q-1)b},\qquad
r=d\lambda^2=\frac{(a-b)^2}{q\,(a+(q-1)b)} .
\label{eq:dlr}
\end{equation}
The map $(a,b)\mapsto(d,\lambda)$ is invertible, $a=d(1+(q-1)\lambda)$ and $b=d(1-\lambda)$, with $\lambda\in[-1/(q-1),1]\setminus\{0\}$; $\lambda<0$ in the disassortative case and $\lambda=-1/(q-1)$ for planted coloring. The condition $d\lambda^2>1$ is equivalent to $|a-b|>q\sqrt d$. The surface $d\lambda^2=1$ is the \emph{Kesten--Stigum threshold}.
\end{definition}

\begin{definition}[Overlap, weak recovery, detection]
\label{def:overlap}
For labelings $\hat\sigma,\sigma\in\Lab$,
\begin{equation}
\overlap(\hat\sigma,\sigma)=\max_{\pi\in S_q}\ \frac1n\,\#\{i:\hat\sigma_i=\pi(\sigma_i)\},
\label{eq:overlap}
\end{equation}
the maximum being over the $q!$ permutations of the label names. A random labeling has overlap $1/q+o(1)$. An algorithm achieves \emph{weak recovery} (with margin $\varepsilon$) if its output satisfies $\overlap(\hat\sigma,\sigma)\ge1/q+\varepsilon$ with probability $1-o(1)$ for a fixed $\varepsilon>0$. Exact recovery is impossible in the sparse regime, since a fraction $e^{-d}+o(1)$ of the vertices are isolated. \emph{Detection} is the task of distinguishing $\mathrm{SBM}(n,q,a,b)$ from the Erd\H{o}s--R\'enyi graph with the same average degree, $\mathrm{SBM}(n,q,d,d)$, with probability of error $o(1)$.
\end{definition}

\subsection{Computational classes}

An algorithm is \emph{polynomial-time} (efficient, fast) if its running time is bounded by a fixed power of $n$; the exhaustive search over the $q^n$ labelings is not. A \emph{low-degree rule} of degree $D$ is a map $A\mapsto t(A)$ each of whose coordinates is a polynomial of degree at most $D$ in the entries of $A$; $\cP_D$ denotes the set of such rules. Counts of edges, paths, cycles and other subgraphs with at most $D$ edges, weighted sums of walk counts of length at most $D$, $D$ steps of a power iteration (repeated multiplication of a vector by $A$) and $D$ rounds of linearized message passing (belief propagation with its update replaced by its derivative at the uninformative fixed point) all belong to $\cP_D$. Polynomials of degree $D=O(\log n)$ are believed to capture the power of polynomial-time algorithms for problems of this type, and lower bounds against $\cP_D$ are the standard evidence of computational hardness \citep{hopkins2017,kunisky2019,schramm2022}. Let $\R[A]_{\le D}$ denote the space of polynomials of degree at most $D$ in the entries of $A$. For a prior $\xi$ on labelings, the \emph{low-degree correlation} of \citet{schramm2022,chin2025} is
\begin{equation}
\Corr_D(\xi)=\sup_{f\in\R[A]_{\le D}}\ \frac{\E_\xi[f(A)\,x_{12}]}{\sqrt{\E_\xi[f(A)^2]\,\E_\xi[x_{12}^2]}},\qquad x_{12}=\ind{\sigma_1=\sigma_2}-\tfrac1q ,
\label{eq:corr}
\end{equation}
the expectations being over $(\sigma,A)$ with $\sigma\sim\xi$. We write $\Corr_D$ for the i.i.d.\ uniform prior.

\subsection{Decision theory}

A decision problem consists of a parameter $\theta\in\Theta$, data $X\sim P_\theta$, an action space, and a loss $L(\theta,t)\ge0$. A \emph{decision rule} $\delta$ maps data to actions; a \emph{randomized} rule may use auxiliary randomness. The \emph{risk} is $\risk(\theta,\delta)=\E_\theta L(\theta,\delta(X))$, the \emph{Bayes risk} under a prior $\xi$ is $\brisk(\xi,\delta)=\int\risk(\theta,\delta)\,\xi(\dt\theta)$, and a \emph{Bayes rule} for $\xi$ minimizes it. A rule is \emph{minimax} if it minimizes $\sup_\theta\risk(\theta,\delta)$; a prior is \emph{least favorable} if its Bayes risk $\inf_\delta\brisk(\xi,\delta)$ is maximal among priors. A rule $\delta$ is \emph{dominated} by $\delta'$ if $\risk(\theta,\delta')\le\risk(\theta,\delta)$ for all $\theta$ with strict inequality for some $\theta$, and \emph{admissible} if it is not dominated. A rule is \emph{equivariant} (under vertex relabeling) if $\delta(gA)=g\delta(A)$ for every permutation $g$. We use the following classical facts \citep[Chapters~1, 4 and~5]{lehmann1998}: Jensen's inequality $f(\E Z)\le\E f(Z)$ for convex $f$; a Bayes rule with constant risk is minimax and its prior is least favorable; for a strictly convex loss the Bayes rule is unique; a unique Bayes rule is admissible, and so is every Bayes rule for a prior with positive mass on each point of a finite $\Theta$; for a convex loss, every randomized rule is matched or improved by a nonrandomized one. The Cauchy--Schwarz inequality is used in the forms $(\sum a_ib_i)^2\le\sum a_i^2\sum b_i^2$ and $\Cov(X,Y)^2\le\Var X\,\Var Y$.

\subsection{Estimation}

An estimator $T$ of $g(\theta)$ is unbiased if $\E_\theta T=g(\theta)$ for all $\theta$. The score is $S=\partial_\theta\log P_\theta(X)$, the Fisher information is $I(\theta)=\Var_\theta(S)$, and the Cram\'er--Rao bound states that $\Var_\theta(T)\ge g'(\theta)^2/I(\theta)$ for every unbiased $T$. A UMVUE has minimum variance among unbiased estimators at every $\theta$; a BLUE has minimum variance among unbiased estimators that are affine functions of the observations. A statistic is sufficient if the conditional law of the data given it does not depend on $\theta$, and complete if no nonzero function of it has mean zero for all $\theta$. A one-parameter exponential family has densities $\beta(\theta)e^{\theta T(x)}h(x)$, $\theta$ ranging over the natural parameter space $\{\theta:\int e^{\theta T}h<\infty\}$, assumed to contain an open interval; $T$ is then complete and sufficient, and by the Lehmann--Scheff\'e theorem \citep{lehmann1950} any unbiased function of $T$ is the unique UMVUE of its expectation. Karlin's theorem \citep{karlin1958} is stated in Section~\ref{sec:karlin}.

\subsection{Probability}

A Poisson variable with mean $\mu$ has $\Prob(X=x)=e^{-\mu}\mu^x/x!$ and variance $\mu$. For laws $p,p_0$ on a countable set, the Hellinger affinity is $\rho=\sum_x\sqrt{p(x)p_0(x)}\in[0,1]$ and the total variation distance is $\TV=\frac12\sum_x|p(x)-p_0(x)|$; the minimal error probability of a test between two equiprobable hypotheses is $\frac12(1-\TV)$. A Markov chain on a finite set is specified by its matrix of transition probabilities.

\subsection{Algorithms}

Belief propagation (BP) is the message-passing algorithm of \citet{decelle2011}, written out in Section~\ref{sec:bp}; its \emph{uninformative} fixed point assigns probability $1/q$ to every label in every message. Spectral clustering embeds the vertices by $q$ eigenvectors of a matrix built from $A$ and partitions the embedding by $k$-means; we use the Bethe Hessian $H(\varrho)=(\varrho^2-1)I-\varrho A+\Diag(\deg)$ with $\varrho=\sgn(\lambda)\sqrt d$ \citep{saade2014}, and the adjacency matrix for comparison. Label propagation is nearest-neighbor voting on the graph with a fraction of labels revealed. Gibbs sampling is the Markov chain on labelings with the posterior as stationary law. Path counting scores a pair of vertices by a weighted count of the non-backtracking walks between them. PageRank is the stationary law of the random walk with restart probability $1-\alpha$ at a uniform vertex, and personalized PageRank restarts at a fixed vertex. The EM algorithm \citep{dempster1977} is the standard maximum-likelihood procedure for incomplete data. All are specified in Appendix~\ref{app:algorithms}.

\section{The three model numbers and the threshold}
\label{sec:threshold}

This section derives $d$, $\lambda$ and the role of $d\lambda^2$ from first principles, fixes a running example, and quantifies the combinatorial barrier that separates exhaustive search from efficient computation. The material is classical \citep{kesten1966,evans2000,decelle2011,krzakala2013}; we include it because the later sections build on the explicit formulas.

\subsection{Average degree and signal strength}
\label{sec:lambda}

Fix a vertex. About $n/q$ other vertices share its label, each adjacent to it with probability $a/n$, and about $(q-1)n/q$ do not, each adjacent with probability $b/n$; hence the expected degree is $a/q+(q-1)b/q=d$. In the sparse regime the degree is asymptotically Poisson with mean $d$, and a fraction $e^{-d}$ of the vertices are isolated.

Condition on an edge $\{i,j\}$ and on $\sigma_i=s$. Since edges are drawn with probabilities proportional to $a$ or $b$, the conditional law of $\sigma_j$ puts mass $a/(a+(q-1)b)$ on $s$ and $b/(a+(q-1)b)$ on each other label. The $q\times q$ matrix $M$ of these conditional probabilities is a Markov transition matrix, the \emph{label chain}, and with $J$ the all-ones matrix,
\begin{equation}
M=\frac{bJ+(a-b)I}{a+(q-1)b}=\frac Jq+\lambda\Big(I-\frac Jq\Big).
\label{eq:M}
\end{equation}
The all-ones vector is an eigenvector with eigenvalue $1$, and every vector with zero sum is an eigenvector with eigenvalue $(a-b)/(a+(q-1)b)=\lambda$, of multiplicity $q-1$. Since $J/q$ and $I-J/q$ are complementary projections,
\begin{equation}
M^k=\frac Jq+\lambda^k\Big(I-\frac Jq\Big),\qquad
\Prob(\sigma_v=\sigma_u\mid v\text{ at distance $k$ from $u$})=\frac{1+(q-1)\lambda^k}{q},
\label{eq:Mk}
\end{equation}
the probability referring to the two ends of a path of length $k$. Thus $\lambda$ measures how much one edge reveals about the label at its far end, and $\lambda^k$ how much a path of length $k$ reveals.

\subsection{\texorpdfstring{Why $d\lambda^2$}{Why d lambda squared}}
\label{sec:tree}

In the sparse regime the ball of radius $o(\log n)$ around a vertex is, with high probability, a tree in which each vertex has a Poisson($d$) number of children and labels are transmitted along edges by the chain $M$. Consider $q=2$, write labels as $\pm1$, let $X_0$ be the root label and let $Y_k$ be the sum of the labels at depth $k$. By \eqref{eq:Mk}, $\E[X_v\mid X_0]=\lambda^kX_0$ for a vertex $v$ at depth $k$, so $\E[Y_k\mid X_0]=(d\lambda)^kX_0$. Two vertices at depth $k$ whose most recent common ancestor lies at depth $j\ge1$ satisfy $\E[X_uX_v\mid X_0]=\lambda^{2(k-j)}$, and there are $d^{2k-j}$ such ordered pairs on average, while pairs whose common ancestor is the root are conditionally independent given $X_0$. Writing $\ell=k-j$ and $r=d\lambda^2$, the conditional variance is, up to the centering terms (which cancel against the fluctuation of the generation size on a Poisson tree),
\begin{equation*}
\Var(Y_k\mid X_0)\approx d^k\sum_{\ell=0}^{k-1}r^\ell=d^k\,\frac{r^k-1}{r-1},
\end{equation*}
and the signal-to-noise ratio of the depth-$k$ statistic is
\begin{equation}
\mathrm{SNR}_k^2=\frac{\E[Y_k\mid X_0]^2}{\Var(Y_k\mid X_0)}\approx\frac{r^k(r-1)}{r^k-1}
\ \longrightarrow\
\begin{cases}
r-1 & (r>1),\\
0\ \text{like }1/k & (r=1),\\
0\ \text{like }r^k & (r<1).
\end{cases}
\label{eq:snr}
\end{equation}
Above the threshold the information carried by the vertices at distance $k$ does not decay with $k$; below it, the signal-to-noise ratio of the depth-$k$ census decays geometrically. This is the Kesten--Stigum threshold, and \eqref{eq:snr} is the mechanism by which many weak clues, each of strength $\lambda^k$, combine into a usable one when $d\lambda^2>1$. Two remarks on what \eqref{eq:snr} does and does not say are needed later. First, it concerns the linear statistic $Y_k$, the census of the labels at distance $k$, and the threshold it locates is exact only for that statistic. For $q=2$ it is also the reconstruction threshold of the tree: below it the information that the labels at distance $k$ carry about the root tends to zero as $k\to\infty$, whatever function of those labels is used \citep{evans2000}. For $q\ge5$ it is not: for every degree, nonlinear functions of the labels at distance $k$ retain information about the root as $k\to\infty$ at values of $d\lambda^2$ strictly below $1$ \citep{sly2011,mossel2023}, and correspondingly the belief-propagation recursion has an informative fixed point there. Second, \eqref{eq:snr} describes inference started from ignorance about the labels. The informative fixed point is a statement about self-consistency: a labeling that is already correct is preserved by local updates. The hard window of Section~\ref{sec:decision} lies between the two: the true labeling is a locally stable fixed point of belief propagation while every local computation started from an uninformative state stays at chance, which is exactly what the experiments of Section~\ref{sec:expq5} display. The same number appears in the spectrum of the non-backtracking matrix of the graph: its bulk eigenvalues have modulus about $\sqrt d$ and its $q-1$ informative eigenvalues are $d\lambda$, which separate from the bulk exactly when $d\lambda>\sqrt d$ \citep{krzakala2013,bordenave2018}.

\subsection{A running example}
\label{sec:example}

Let $q=5$, $a=3$, $b=1$. Then
\begin{equation*}
d=\frac{3+4}{5}=1.4,\qquad\lambda=\frac{3-1}{3+4}=\frac27,\qquad r=d\lambda^2=\frac75\cdot\frac4{49}=\frac4{35}\approx0.114,
\end{equation*}
a factor $8.75$ below the threshold; for $r=1$ with $q=5$ and $b=1$ one needs $(a-1)^2=5(a+4)$, that is $a=(7+\sqrt{125})/2\approx9.09$. A fraction $e^{-1.4}\approx24.7\%$ of the vertices are isolated. All numerical constants of Section~\ref{sec:fisher} are evaluated on this example.

\subsection{The combinatorial barrier}
\label{sec:bruteforce}

There are $q^n$ labelings ($5^{100}\approx7.9\times10^{69}$ for $q=5$ and $n=100$), and the $q!$ equivalent renamings do not change the order of magnitude. Below the information-theoretic threshold the search is also useless: an Erd\H{o}s--R\'enyi graph contains labelings with an atypically large number of intra-community edges, and the maximum-likelihood labeling is one of them, uncorrelated with the truth. Above the information-theoretic threshold but below $d\lambda^2=1$, the difficulty for efficient algorithms is of a different nature, which Section~\ref{sec:bp} makes precise: by \eqref{eq:snr} the census of every bounded neighborhood carries a signal that vanishes with its depth, so local updates started from an uninformative state have no gradient toward the truth. The posterior landscape then has two basins, an uninformative one that is locally stable until $d\lambda^2$ reaches~1 and an informative one around the true labeling, and local algorithms start in the wrong basin \citep{decelle2011,zdeborova2016}.

\section{Decision theory: the gap as a minimax theorem}
\label{sec:decision}

\subsection{The recovery problem as a decision problem}
\label{sec:setting}

The parameter is the labeling $\sigma\in\Theta=\Lab$, the data is $A\sim P_\sigma$, and the target is the pair array
\begin{equation}
x_{ij}(\sigma)=\ind{\sigma_i=\sigma_j}-\frac1q,\qquad i<j .
\label{eq:target}
\end{equation}
An action is a real array $t=(t_{ij})_{i<j}$, and the loss is the mean squared error over the $N=\binom n2$ pairs,
\begin{equation}
L(\sigma,t)=\frac1N\sum_{i<j}\big(t_{ij}-x_{ij}(\sigma)\big)^2 ,
\label{eq:loss}
\end{equation}
which is strictly convex in $t$ and invariant under renaming of the communities. The risk is $\risk(\sigma,\delta)=\E_\sigma L(\sigma,\delta(A))$. Two priors are used: the i.i.d.\ uniform prior $\pi$, which is the uniform law on the $q^n$ labelings, and for each profile $m$ the uniform law $\pi_m$ on the orbit $\Theta_m$; $\pi$ is the mixture of the $\pi_m$ with the multinomial weights $\pi(\Theta_m)$. We write $\brisk(\xi,\delta)$ for the Bayes risk under a prior $\xi$ and $\brisk(\delta)=\brisk(\pi,\delta)$.

The trivial rule $t\equiv0$ has risk $\triv(\sigma)=N^{-1}\sum_{i<j}x_{ij}(\sigma)^2$, which equals $\vq=\frac1q(1-\frac1q)$ on average under $\pi$ and $\vq+O(n^{-1/3})$ uniformly on the set
\begin{equation}
\Lam=\big\{\sigma:\ |m_s(\sigma)/n-1/q|\le n^{-1/3}\ \text{for all } s\big\},\qquad \pi(\Lam)\to1\ \ \text{(Chebyshev)},
\label{eq:typical}
\end{equation}
of labelings with typical profile; on an orbit $\Theta_m$ the best constant rule is $t\equiv\E_{\pi_m}x_{12}$, with risk $\Var_{\pi_m}(x_{12})$. A rule is useful at $\sigma$ only if its risk is below $\triv(\sigma)$. The symmetry of the problem is expressed by
\begin{equation}
P_\sigma(gA\in\cdot)=P_{g\sigma}(\cdot),\qquad L(g\sigma,gt)=L(\sigma,t),\qquad \Theta_m=\{g\sigma: g\in S_n\}\ \text{for any } \sigma\in\Theta_m ,
\label{eq:symmetry}
\end{equation}
for every permutation $g$ of the vertices: no labeling of a given profile is distinguished.

\subsection{Worst case equals average case}

\begin{lemma}[Symmetrization]
\label{lem:sym}
For a rule $\delta$ let $\bar\delta(A)=\frac1{n!}\sum_{g\in S_n}g^{-1}\delta(gA)$. Then $\bar\delta$ is equivariant, its risk is constant on each orbit $\Theta_m$, and
\begin{equation*}
\risk(\sigma,\bar\delta)\ \le\ \brisk(\pi_{m(\sigma)},\delta)\qquad\text{for every }\sigma .
\end{equation*}
\end{lemma}

\begin{proof}
By Jensen's inequality and the convexity of $L$ in $t$, $L(\sigma,\bar\delta(A))\le\frac1{n!}\sum_gL(\sigma,g^{-1}\delta(gA))=\frac1{n!}\sum_gL(g\sigma,\delta(gA))$, using the invariance of the loss in \eqref{eq:symmetry}. Taking $\E_\sigma$ and using $P_\sigma(gA\in\cdot)=P_{g\sigma}$ gives $\risk(\sigma,\bar\delta)\le\frac1{n!}\sum_g\risk(g\sigma,\delta)$. As $g$ ranges over $S_n$, $g\sigma$ visits each element of $\Theta_{m(\sigma)}$ equally often, so the right-hand side is $\brisk(\pi_{m(\sigma)},\delta)$. Equivariance, $\bar\delta(hA)=h\bar\delta(A)$, holds because $h$ permutes the terms of the sum; it implies $\risk(h\sigma,\bar\delta)=\E_\sigma L(h\sigma,\bar\delta(hA))=\E_\sigma L(h\sigma,h\bar\delta(A))=\risk(\sigma,\bar\delta)$, hence constant risk on orbits.
\end{proof}

\begin{theorem}[Worst case equals average case]
\label{thm:worst}
\begin{enumerate}[label=\textup{(\alph*)},leftmargin=2em]
\item Let $\cD$ be a class of rules closed under symmetrization, that is, $\bar\delta\in\cD$ whenever $\delta\in\cD$; this holds in particular if $\cD$ is closed under averaging and under relabeling $\delta\mapsto g^{-1}\delta(g\,\cdot)$. Then for every profile $m$,
\begin{equation*}
\inf_{\delta\in\cD}\ \sup_{\sigma\in\Theta_m}\risk(\sigma,\delta)=\inf_{\delta\in\cD}\ \brisk(\pi_m,\delta).
\end{equation*}
\item Let $\delta$ be any rule and $L$ any loss invariant under relabeling. The randomized rule that draws $g$ uniformly from $S_n$ and outputs $g^{-1}\delta(gA)$ has risk $\brisk(\pi_{m(\sigma)},\delta)$ at every $\sigma$, at the cost of one pass over the graph.
\end{enumerate}
\end{theorem}

\begin{proof}
(a) A supremum dominates an average, so the left side is at least the right side. Conversely, for $\delta\in\cD$ the symmetrization $\bar\delta$ lies in $\cD$ and, by Lemma~\ref{lem:sym}, $\sup_{\Theta_m}\risk(\cdot,\bar\delta)\le\brisk(\pi_m,\delta)$; take the infimum over $\delta$. (b) The risk of the randomized rule at $\sigma$ is $\frac1{n!}\sum_g\risk(g\sigma,\delta)=\brisk(\pi_{m(\sigma)},\delta)$, as in the proof of Lemma~\ref{lem:sym}; convexity is not needed.
\end{proof}

Part (a) is the minimax theorem of a constant-risk Bayes rule, applied \emph{inside} an arbitrary symmetric convex class: the uniform prior on each profile is least favorable for the restricted statistician as well as for the unrestricted one. Part (b) shows that worst-case and average-case hardness are the same question in this problem for every algorithm, since any algorithm that is good on average on a profile can be made equally good on every labeling of that profile by one random relabeling of its input. The Bayesian (super-population) and the fixed-labeling (minimax) viewpoints therefore coincide. Lemma~\ref{lem:sym} also exhibits a computational cost hidden in the classical statement that, for a convex loss, a randomized rule can be replaced by a nonrandomized one that is at least as good: the replacement costs $n!$ evaluations of $\delta$.

\subsection{The posterior mean}

The posterior mean under $\pi$ is
\begin{equation}
\delta^*(A)_{ij}=\E_\pi[x_{ij}\mid A]=\Prob_\pi(\sigma_i=\sigma_j\mid A)-\frac1q ,
\label{eq:postmean}
\end{equation}
a sum over all $q^n$ labelings; for the balanced orbit we also use $\delta^*_{\mathrm{bal}}$, the posterior mean under $\pi_{\mathrm{bal}}$, a sum over the $n!/((n/q)!^q)$ balanced labelings. Both are the precise form of ``try every labeling''. We assume $a,b>0$ in the next theorem, so that every graph has positive probability under every labeling.

\begin{theorem}[Optimality of the posterior mean]
\label{thm:bayes}
\begin{enumerate}[label=\textup{(\roman*)},leftmargin=2em]
\item $\delta^*$ is the unique Bayes rule for $\pi$ and is admissible; its risk depends on $\sigma$ only through the profile $m(\sigma)$.
\item When $q\mid n$, $\delta^*_{\mathrm{bal}}$ is the unique Bayes rule for $\pi_{\mathrm{bal}}$, has constant risk on $\Theta_{\mathrm{bal}}$, is the unique minimax rule for the parameter set $\Theta_{\mathrm{bal}}$, and is admissible there; $\pi_{\mathrm{bal}}$ is least favorable.
\end{enumerate}
\end{theorem}

\begin{proof}
For a fixed $A$ the posterior expected loss of an action is $N^{-1}\sum_{i<j}\E[(t_{ij}-x_{ij})^2\mid A]$, each term strictly convex in $t_{ij}$ with unique minimizer $\E[x_{ij}\mid A]$; a rule that minimizes the posterior expected loss for every $A$ is Bayes, and strict convexity makes it unique. Under either prior the model is invariant under relabeling, so the posterior mean is equivariant and its risk is constant on orbits (Lemma~\ref{lem:sym}). Admissibility of $\delta^*$ follows because $\Theta$ is finite and $\pi$ charges every point. For (ii), a Bayes rule with constant risk is minimax and its prior is least favorable; if $\delta'$ were another minimax rule for $\Theta_{\mathrm{bal}}$ then $\brisk(\pi_{\mathrm{bal}},\delta')\le\sup_{\Theta_{\mathrm{bal}}}\risk(\cdot,\delta')\le\sup_{\Theta_{\mathrm{bal}}}\risk(\cdot,\delta^*_{\mathrm{bal}})=\brisk(\pi_{\mathrm{bal}},\delta^*_{\mathrm{bal}})$, so $\delta'$ is Bayes for $\pi_{\mathrm{bal}}$ and equals $\delta^*_{\mathrm{bal}}$; a unique Bayes rule is admissible.
\end{proof}

\begin{lemma}[A posterior sample costs a factor two]
\label{lem:sample}
Let $\tilde\sigma$ be drawn from the posterior of $\sigma$ given $A$ under $\pi$ (respectively $\pi_{\mathrm{bal}}$), and $\tilde x=x(\tilde\sigma)$. Then $\brisk(\pi,\tilde x)=2\,\brisk(\pi,\delta^*)$ (respectively with $\pi_{\mathrm{bal}}$ and $\delta^*_{\mathrm{bal}}$).
\end{lemma}

\begin{proof}
Given $A$, $\tilde x_{ij}$ and $x_{ij}$ are independent with the same law, so $\E[(\tilde x_{ij}-x_{ij})^2\mid A]=2\Var(x_{ij}\mid A)$, whereas $\E[(\delta^*_{ij}-x_{ij})^2\mid A]=\Var(x_{ij}\mid A)$. Average over pairs and over $A$.
\end{proof}

An ideal Gibbs sampler returns a posterior sample, and Lemma~\ref{lem:sample} quantifies the price of reporting one sample instead of the posterior mean: exactly a factor two in Bayes risk, at every signal strength.

\subsection{The low-degree class}
\label{sec:gap}

The class $\cP_D$ is a linear space, hence closed under averaging, and relabeling maps a polynomial of degree $D$ to a polynomial of degree $D$, so Theorem~\ref{thm:worst}(a) applies to it. (PageRank-type scores divide by vertex degrees; they belong to $\cP_D$ once the degrees are replaced by their mean $d$.)

\begin{theorem}[Bayes and minimax risk of the low-degree class]
\label{thm:lowdeg}
For every prior $\xi$ invariant under vertex relabeling, in particular for $\pi$ and for each $\pi_m$,
\begin{equation}
\inf_{\delta\in\cP_D}\brisk(\xi,\delta)=\E_\xi[x_{12}^2]\,\big(1-\Corr_D(\xi)^2\big),
\label{eq:lowdeg}
\end{equation}
with $\Corr_D(\xi)$ as in \eqref{eq:corr}. For $\xi=\pi$ the right-hand side is $\vq(1-\Corr_D^2)$, and by Theorem~\ref{thm:worst}(a) the minimax risk of $\cP_D$ over $\Theta_m$ equals \eqref{eq:lowdeg} with $\xi=\pi_m$.
\end{theorem}

\begin{proof}
Under a relabeling-invariant prior every pair of vertices plays the same role, and the coordinates $t_{ij}$ of a rule in $\cP_D$ may be chosen independently, so $\inf_{\cP_D}\brisk(\xi,\cdot)=\inf_{\deg f\le D}\E_\xi[(f(A)-x_{12})^2]$. For fixed $f$ the quadratic $s\mapsto\E[(sf-x_{12})^2]=s^2\E f^2-2s\E[fx_{12}]+\E x_{12}^2$ is minimized at $s=\E[fx_{12}]/\E f^2$ with value $\E x_{12}^2\,(1-\E[fx_{12}]^2/(\E f^2\,\E x_{12}^2))$; optimizing over $f$ gives \eqref{eq:lowdeg}. Under $\pi$, $\Prob(\sigma_1=\sigma_2)=1/q$ and $\E x_{12}^2=\vq$.
\end{proof}

Theorem~\ref{thm:lowdeg} shows that the low-degree correlation, an average-case quantity, determines the Bayes and the minimax risk of the low-degree class exactly. The two external inputs are the following.

\begin{fact}[Low-degree bound, \citet{chin2025}, Theorem~1.3]
\label{fact:chin}
Suppose $q\le n^{1/2-\gamma_0}$, $a\ge c_1b$ and $r\le1-\eta$ for constants $\gamma_0,c_1,\eta>0$. Then there are constants $c_0=c_0(c_1,\eta)$ and $C=C(\eta)$ such that $\Corr_D\le C\sqrt{q/n}$ for all $D\le n^{c_0\gamma_0}$.
\end{fact}

\begin{fact}[Recovery below the threshold, \citet{banks2016,abbe2018,mossel2023}]
\label{fact:info}
For every $q\ge5$ there are parameters $(a,b)$ with $a>0$ and $r<1$, and an $\varepsilon>0$, such that some (exponential-time) algorithm achieves $\overlap(\hat\sigma,\sigma)\ge1/q+\varepsilon$ with probability $1-o(1)$ under $\mathrm{SBM}(n,q,a,b)$.
\end{fact}

\citet{banks2016} establish Fact~\ref{fact:info} in the disassortative case for $q\ge5$ and in the assortative case for $q\ge11$, through explicit bounds on the detection region that are continuous in $(a,b)$, so that the window contains parameters with $a>0$; \citet{mossel2023} prove that the Kesten--Stigum bound is not the information-theoretic threshold for any $q\ge5$. The hypotheses of Facts~\ref{fact:chin} and~\ref{fact:info} are therefore simultaneously satisfiable for every $q\ge5$, and we call the set of such parameters the \emph{hard window}.

\begin{lemma}[Overlap controls risk]
\label{lem:chance}
Let $\sigma\in\Lam$ and let $\hat\sigma$ be a labeling with $\overlap(\hat\sigma,\sigma)\ge1/q+\varepsilon$. Let $u_{ij}=\ind{\hat\sigma_i=\hat\sigma_j}-\bar u$, where $\bar u$ is the mean of $\ind{\hat\sigma_i=\hat\sigma_j}$ over pairs, and $t=wu$ with $w=\varepsilon^2/q$. Then $|t_{ij}|\le1$ and
\begin{equation*}
L(\sigma,t)\le\triv(\sigma)-\frac{\varepsilon^4}{q^2}+O(n^{-1/3}),
\end{equation*}
the $O(n^{-1/3})$ term being uniform in $\sigma\in\Lam$ and $\hat\sigma$.
\end{lemma}

\begin{proof}
Rename the labels of $\hat\sigma$ so that the overlap is attained, which does not change $u$, and let $n_{st}=\#\{i:\sigma_i=s,\hat\sigma_i=t\}$ and $\hat m_t=\sum_sn_{st}$. Since $\sum_{i<j}u_{ij}=0$, the mean over pairs of $u_{ij}x_{ij}$ equals the mean of $u_{ij}\ind{\sigma_i=\sigma_j}$, which, up to $O(n^{-1/3})$ because $\sigma\in\Lam$, is
\begin{align*}
K&=\frac1{n^2}\Big[\sum_{s,t}n_{st}^2-\frac1q\sum_t\hat m_t^2\Big]=\frac1{n^2}\sum_t\sum_s\Big(n_{st}-\frac{\hat m_t}q\Big)^2\\
&\ge\ \frac1{n^2}\sum_s\Big(n_{ss}-\frac{\hat m_s}q\Big)^2\ \ge\ \frac1{qn^2}\Big(\sum_sn_{ss}-\frac nq\Big)^2\ \ge\ \frac{\varepsilon^2}q ,
\end{align*}
by the identity $\sum_s(n_{st}-\hat m_t/q)^2=\sum_sn_{st}^2-\hat m_t^2/q$, by keeping the diagonal terms, by Cauchy--Schwarz with $\sum_s\hat m_s=n$, and by $\sum_sn_{ss}\ge(1/q+\varepsilon)n$. Hence $L(\sigma,wu)=w^2\,\mathrm{avg}(u^2)-2wK+\triv(\sigma)\le w^2-2w\varepsilon^2/q+\triv(\sigma)+O(n^{-1/3})$, since $u^2\le1$, and $w=\varepsilon^2/q$ gives the claim.
\end{proof}

\begin{corollary}[The information--computation gap as a minimax statement]
\label{cor:gap}
Fix $q\ge5$ and parameters in the hard window. There is a sequence of sets $G_n\subset\Lab$ with $\pi(G_n)\to1$ such that:
\begin{enumerate}[label=\textup{(\alph*)},leftmargin=2em]
\item every $\delta\in\cP_D$ with $D\le n^{c_0\gamma_0}$ and outputs in $[-1,1]$ satisfies
\begin{equation*}
\sup_{\sigma\in G_n}\risk(\sigma,\delta)\ \ge\ \brisk(\delta)\ -\ o(1)\ \ge\ \vq\,\Big(1-\frac{C^2q}{n}\Big)-o(1);
\end{equation*}
\item there is a (randomized, exponential-time) rule $\delta_{\mathrm{slow}}$ with
\begin{equation*}
\risk(\sigma,\delta_{\mathrm{slow}})\ \le\ \vq-\frac{\varepsilon^4}{q^2}+o(1)\qquad\text{for every }\sigma\in G_n ,
\end{equation*}
and consequently $\brisk(\delta^*)\le\brisk(\delta_{\mathrm{slow}})\le\vq-\varepsilon^4/q^2+o(1)$.
\end{enumerate}
In words: in the hard window no low-degree rule has worst-case risk below the trivial risk $\vq-o(1)$, so the minimax risk of the low-degree class converges to $\vq$, while an exponential-time rule beats $\vq$ by the fixed margin $\varepsilon^4/q^2$ uniformly over a set of labelings of probability tending to one, and the Bayes risk of the posterior mean is smaller than $\vq$ by the same margin.
\end{corollary}

\begin{proof}
Let $F$ be the event that the algorithm of Fact~\ref{fact:info} fails to reach overlap $1/q+\varepsilon$, so that $p_n=\Prob(F)\to0$, and for a profile $m$ let $p_n(m)=\Prob(F\mid\sigma\in\Theta_m)$. Set $G_n=\Lam\cap\{\sigma:p_n(m(\sigma))\le\sqrt{p_n}\}$; since $p_n=\E_\pi[p_n(m(\sigma))]$, Markov's inequality gives $\pi(G_n^c)\le\pi(\Lam^c)+\sqrt{p_n}\to0$. (a) For $\delta$ with outputs in $[-1,1]$ the loss is at most $4$, so $\brisk(\delta)=\E_\pi\risk(\sigma,\delta)\le\sup_{G_n}\risk(\cdot,\delta)+4\pi(G_n^c)$, and Theorem~\ref{thm:lowdeg} with Fact~\ref{fact:chin} bounds $\brisk(\delta)$ from below. (b) Let $\delta_1$ be the rule of Lemma~\ref{lem:chance} applied to the output of the algorithm of Fact~\ref{fact:info}, and $\delta_{\mathrm{slow}}$ its random relabeling, Theorem~\ref{thm:worst}(b), whose risk at $\sigma$ is the Bayes risk of $\delta_1$ under $\pi_{m(\sigma)}$. For $\sigma\in G_n$ the loss of $\delta_1$ is at most $\triv-\varepsilon^4/q^2+O(n^{-1/3})$ on $F^c$, with $\triv=\vq+O(n^{-1/3})$ on $\Lam$, and at most $4$ on $F$; averaging under $\pi_{m(\sigma)}$ gives $\risk(\sigma,\delta_{\mathrm{slow}})\le\vq-\varepsilon^4/q^2+4\sqrt{p_n}+O(n^{-1/3})$. The last claim follows since $\delta^*$ minimizes $\brisk$ and $\brisk(\delta_{\mathrm{slow}})\le\sup_{G_n}\risk(\cdot,\delta_{\mathrm{slow}})+4\pi(G_n^c)$.
\end{proof}

\subsection{The choice of loss}
\label{sec:lossremarks}

The pair loss \eqref{eq:loss} is the natural convex loss for recovery up to renaming. The exhaustive search of Section~\ref{sec:bruteforce} is optimal for a different loss.

\begin{proposition}[Maximum likelihood is Bayes for the all-or-nothing loss]
\label{prop:map}
Under the loss $\ind{\hat\sigma\ne\sigma}$ on labelings, the maximum-likelihood labeling is the Bayes rule for $\pi$; under the loss $\ind{\text{the partition induced by $\hat\sigma$ differs from that of $\sigma$}}$ restricted to partitions into $q$ nonempty classes, the maximum-likelihood partition is the Bayes rule. With ties broken uniformly at random the rule is equivariant, has constant risk on every orbit, and is admissible. Under the pair loss \eqref{eq:loss} the Bayes rule is the posterior mean \eqref{eq:postmean}.
\end{proposition}

\begin{proof}
Reporting $\hat\sigma$ has posterior expected loss $1-\Prob(\sigma=\hat\sigma\mid A)$, minimized by the posterior mode, which is the maximum-likelihood labeling because $\pi$ is uniform on $\Lab$; a partition into $q$ nonempty classes is induced by exactly $q!$ labelings, so the posterior mode over such partitions is the maximum-likelihood partition. Random tie-breaking makes the rule equivariant, hence of constant risk on orbits by Lemma~\ref{lem:sym}; a Bayes rule for a prior charging every point of a finite parameter set is admissible. The last sentence is Theorem~\ref{thm:bayes}(i).
\end{proof}

\begin{lemma}[Clipping]
\label{lem:clip}
Since $x_{ij}\in[-1/q,1-1/q]$, replacing each output $t_{ij}$ of a rule by its projection onto $[-1/q,1-1/q]$ does not increase the risk at any $\sigma$, and decreases it strictly whenever the rule leaves the interval with positive probability.
\end{lemma}

\begin{proof}
Projection onto an interval containing $x$ does not increase $|t-x|$, and decreases it when $t$ lies outside.
\end{proof}

A clipped polynomial is not, in general, a polynomial of degree at most $D$, so every low-degree rule that leaves the interval with positive probability is dominated by a rule outside $\cP_D$; in the exact computations of Section~\ref{sec:expexact} a least-squares quadratic rule leaves the interval for $37\%$ of its outputs and clipping reduces its risk by $10\%$.

\section{Belief propagation and empirical Bayes}
\label{sec:bp}

\subsection{Linear stability of the uninformative fixed point}

Belief propagation \citep{decelle2011} maintains, for every directed edge $i\to j$ and every label $t\in[q]$, a message $\psi(i\to j,t)$, the estimated probability that $\sigma_i=t$ computed without the information coming from $j$. One iteration is
\begin{equation}
\psi^{\mathrm{new}}(i\to j,t)=\frac{1}{Z_{i\to j}}\,e^{-h_t}\prod_{k\in\partial i\setminus j}\big[\,b+(a-b)\,\psi(k\to i,t)\,\big],\qquad
h_t=\frac1n\sum_{k=1}^n\big[\,b+(a-b)\,\psi_k(t)\,\big],
\label{eq:bp}
\end{equation}
where $\partial i\setminus j$ is the set of neighbors of $i$ other than $j$, $Z_{i\to j}$ normalizes the message to a probability vector, $\psi_k$ is the current marginal of vertex $k$, and the field $h_t$ accounts for the non-edges. The uninformative state $\psi\equiv1/q$ is a fixed point of \eqref{eq:bp} for every graph, since every bracket then equals $b+(a-b)/q=d$.

\begin{theorem}[Linear stability]
\label{thm:bp}
Write $\psi(k\to i,t)=1/q+e(k\to i,t)$ with $\sum_te(k\to i,t)=0$ for the messages, $\psi_k(t)=1/q+\epsilon_k(t)$ for the marginals, and let $\bar e(t)=n^{-1}\sum_k\epsilon_k(t)$ be the mean perturbation of the marginals, which is the quantity that enters the field $h_t$. To first order in the perturbations, one iteration of \eqref{eq:bp} acts as
\begin{equation}
e^{\mathrm{new}}(i\to j,t)=\lambda\Big[\sum_{k\in\partial i\setminus j}e(k\to i,t)-d\,\bar e(t)\Big].
\label{eq:bplin}
\end{equation}
Thus the derivative of the belief-propagation map at the uninformative fixed point is $\lambda$ times the non-backtracking operator $B$ of the graph, followed by the subtraction of one and the same vector from every message.
\begin{enumerate}[label=(\alph*),leftmargin=2em]
\item The subtraction acts on the uniform perturbation, which is the leading eigenvector of $B$, with eigenvalue $d$: on average over the degrees of a sparse random graph, a perturbation that is the same on every directed edge is annihilated, to first order, within two iterations, whereas without the field it would be multiplied by $|\lambda|d$ per iteration. On perturbations whose marginals have mean zero over the graph the map is $\lambda B$.
\item If the incoming perturbations are independent with mean zero and variance $v$, the expected variance of the outgoing perturbation on a sparse random graph is $d\lambda^2v$: the norm of a random perturbation is multiplied by $|\lambda|\sqrt d$ per iteration.
\end{enumerate}
\end{theorem}

\begin{proof}
A bracket equals $b+(a-b)(1/q+e)=d+(a-b)e=d(1+q\lambda e)$, since $a-b=q\lambda d$, and the field is $h_t=b+(a-b)(1/q+\bar e(t))=d+q\lambda d\,\bar e(t)$, so that $e^{-h_t}=e^{-d}\big(1-q\lambda d\,\bar e(t)\big)+O(\bar e^2)$. The product over the $K$ neighbors is $d^K\big(1+q\lambda\sum_ke(k\to i,t)\big)+O(e^2)$. The unnormalized message is therefore $e^{-d}d^K\big(1+q\lambda\sum_ke(k\to i,t)-q\lambda d\,\bar e(t)\big)$ up to second-order terms, and its sum over $t$ is $qe^{-d}d^K$ up to second-order terms, because both $e(k\to i,\cdot)$ and $\bar e(\cdot)$ sum to zero over the labels; dividing gives \eqref{eq:bplin}.

For (a), let $e(k\to i,t)=u(t)$ on every directed edge and $\bar e(t)=m(t)$. By \eqref{eq:bplin}, $e^{\mathrm{new}}(i\to j,t)=\lambda(\deg_i-1)u(t)-\lambda d\,m(t)$, and the same computation over the full neighborhood $\partial i$ gives the new marginals $\epsilon^{\mathrm{new}}_i(t)=\lambda\deg_i\,u(t)-\lambda d\,m(t)$. In a sparse random graph the mean degree is $d$ and the mean excess degree, the average of $\deg_i-1$ over directed edges $i\to j$, is $\E[\deg(\deg-1)]/\E\deg=d$ for a Poisson degree law. Averaging, the new messages and the new marginals both have mean perturbation $\lambda d\,(u(t)-m(t))$, so after one iteration $u=m$ and after the next both vanish. Without the field ($m\equiv0$) the uniform perturbation is multiplied by $\lambda d$ at every iteration. Degree fluctuations transfer a part of the uniform perturbation, of the order of the perturbation itself, to the other modes, where (b) applies.

For (b), the field term is the same for all messages and, for independent perturbations with mean zero, has standard deviation $O(\lambda\sqrt{dv/n})$, which is negligible; the variance of a sum of independent terms is the sum of the variances, and a directed edge has on average $d$ outgoing non-backtracking continuations, so the expected new variance is $\lambda^2\cdot d\cdot v$.
\end{proof}

The independence assumption in (b) holds for $o(\log n)$ iterations, while neighborhoods are trees; for longer times the spectral theorem of \citet{bordenave2018} gives the same factor. By that theorem all but $q$ eigenvalues of $B$ have modulus at most $\sqrt d+o(1)$; of the $q$ exceptional ones, the eigenvalue $d$ belongs to the uniform perturbation, which the field neutralizes by (a), and the $q-1$ eigenvalues $d\lambda$ belong to eigenvectors correlated with the communities, on which the map multiplies by $d\lambda^2$ per iteration. Below the threshold $d\lambda^2<|\lambda|\sqrt d$, so the random component (b) decays most slowly and sets the rate; above it the community signal grows faster than the random component and eventually dominates, which is what the growth-factor measurements of Section~\ref{sec:expexact} show (Figure~\ref{fig:checks}B). Theorem~\ref{thm:bp} is the two-basin picture of Section~\ref{sec:bruteforce}: for $d\lambda^2<1$ the uninformative state is attracting, so an algorithm started without information and improving by local updates returns to it; for $d\lambda^2>1$ it is repelling; in the hard window a second, informative fixed point coexists with it (Section~\ref{sec:expq5} exhibits it by initializing at the truth), and the algorithm starts in the wrong basin. The statistical problem is convex in the action and has a unique optimal rule (Theorem~\ref{thm:bayes}); the search over labelings is not convex, and the landscape of the algorithm has more than one basin.

\begin{corollary}[Time to return to the uninformative state]
\label{cor:rate}
With damping $\alpha\in(0,1]$ (new message $=\alpha\times$update $+(1-\alpha)\times$old) the linearized map is $(1-\alpha)I+\alpha\lambda B$ on perturbations whose marginals have mean zero, the uniform component being removed by the field (Theorem~\ref{thm:bp}(a)). Since the bulk spectrum of $B$ lies in the disk of radius $\sqrt d+o(1)$ \citep{bordenave2018} and the community modes decay faster when $d\lambda^2<1$, the asymptotic per-iteration contraction factor of a random perturbation is $f=1-\alpha+\alpha|\lambda|\sqrt d$, and a perturbation of initial size $s_0$ falls below a tolerance $\tau$ after about $\ln(s_0/\tau)/\ln(1/f)$ iterations, a number that diverges as $d\lambda^2\uparrow1$.
\end{corollary}

\subsection{Empirical Bayes: one EM step at the fixed point}
\label{sec:em}

In practice $a$ and $b$ are unknown hyperparameters. The empirical Bayes approach estimates them from the marginal likelihood $\mathcal L(a,b)=\sum_\sigma\pi(\sigma)P_\sigma(A;a,b)$ and plugs the estimates into the posterior. Maximization of $\mathcal L$ is carried out by the EM algorithm with belief propagation in the E-step \citep{decelle2011}: with the current $(a,b)$, iterate \eqref{eq:bp} to convergence; for each edge $e=\{i,j\}$ compute the two-point belief and the posterior probability that $e$ lies inside a community,
\begin{equation}
s_e=\sum_{t}\psi(i\to j,t)\,\psi(j\to i,t),\qquad p_e=\frac{a\,s_e}{a\,s_e+b\,(1-s_e)} ;
\label{eq:em}
\end{equation}
then update
\begin{equation}
a^{\mathrm{new}}=\frac{2q}{n}\sum_ep_e,\qquad b^{\mathrm{new}}=\frac{2q}{n(q-1)}\sum_e(1-p_e),
\label{eq:em2}
\end{equation}
the sums running over the edge set $E$.

\begin{proposition}[EM at the uninformative fixed point]
\label{prop:em}
At the uninformative fixed point one step of \eqref{eq:em}--\eqref{eq:em2} sets $d$ to the observed mean degree $2\nE/n$ and leaves $\lambda$ unchanged, for every graph.
\end{proposition}

\begin{proof}
With $\psi\equiv1/q$, $s_e=1/q$ and $p_e=(a/q)/(a/q+b(q-1)/q)=a/(qd)$ for every edge; hence $\sum_ep_e=\nE\,a/(qd)$ and $\sum_e(1-p_e)=\nE\,(q-1)b/(qd)$, so $a^{\mathrm{new}}=(2\nE/n)\,a/d$ and $b^{\mathrm{new}}=(2\nE/n)\,b/d$. The ratio $a/b$, and with it $\lambda$, is preserved, while $d^{\mathrm{new}}=(a^{\mathrm{new}}+(q-1)b^{\mathrm{new}})/q=2\nE/n$.
\end{proof}

Combined with Theorem~\ref{thm:bp}, Proposition~\ref{prop:em} describes a self-confirming failure: if the initial guess $\lambda_0$ satisfies $d\lambda_0^2<1$, with $d$ the true average degree (it is the graph, not the guess, that supplies the factor $\sqrt d$), belief propagation remains at the uninformative fixed point and $\lambda_0$ is never revised. This is a statement about the algorithm; the corresponding statistical statement is one of identifiability of the mixture $\mathcal L$: for $q=2$, \citet{mossel2015} proved that below the threshold $\mathrm{SBM}(n,2,a,b)$ is contiguous to $\mathrm{SBM}(n,2,d,d)$, so that $\lambda$ cannot be estimated consistently by any method. The experiments of Section~\ref{sec:expexact} show the same phenomenon for $q=3$ and show that an optimistic initialization converges to $d\lambda^2\approx1.06$--$1.09$ whenever the truth is below the threshold.

\section{Fisher information: the threshold as a radius of convergence}
\label{sec:fisher}

The Cram\'er--Rao bound, the theory of unbiased estimation and Karlin's theorem apply to every estimator, efficient or not. They separate the efficient from the inefficient once the Fisher information is computed \emph{within a class of statistics}, and this section carries out that computation at three levels: the complete data (labels known), the observed graph, and the class of functions of short cycle counts. Throughout this section we condition on the community profile, which is independent of $(a,b)$ and therefore carries no information about $\lambda$, and take it balanced for definiteness; ``for large $n$'' means that factors $1-a/n$ and $1-b/n$ are replaced by $1$. Both signs of $\lambda$ are allowed. We assume $a,b>0$, so that $\lambda$ ranges over the open interval $(-1/(q-1),1)$, and the statements below cover the assortative and the disassortative case together through the constant
\begin{equation}
c_\lambda=1-(q-1)\max(0,-\lambda)^3,\qquad c_\lambda\ \le\ 1+(q-1)\lambda^k\ \le\ q\quad(k\ge3),
\label{eq:clambda}
\end{equation}
where the inequalities use $|\lambda|\le1$. Thus $c_\lambda=1$ when $\lambda>0$, $c_\lambda>0$ throughout the range, and $c_\lambda\ge1-(q-1)^{-2}\ge3/4$ when $q\ge3$. We write
\begin{equation*}
\pbar=\frac{a}{qd}=\frac{1+(q-1)\lambda}{q},\qquad
N_{\mathrm{in}}=\frac{n(n/q-1)}{2},\qquad N_{\mathrm{out}}=\frac{n^2(q-1)}{2q},\qquad
\nE=\Ein+\Eout,
\end{equation*}
for the probability that an edge lies inside a community, the numbers of intra- and inter-community vertex pairs, and the numbers of intra-community, inter-community and total edges.

\subsection{Complete data}
\label{sec:known}

Given the labels, $\Ein\sim\mathrm{Bin}(N_{\mathrm{in}},a/n)$ and $\Eout\sim\mathrm{Bin}(N_{\mathrm{out}},b/n)$ are independent and sufficient. The log-likelihood is
\begin{equation*}
\ell=\Ein\ln\tfrac an+(N_{\mathrm{in}}-\Ein)\ln(1-\tfrac an)+\Eout\ln\tfrac bn+(N_{\mathrm{out}}-\Eout)\ln(1-\tfrac bn),
\end{equation*}
with score $\partial_a\ell=(n\,\Ein-a\,N_{\mathrm{in}})/(a(n-a))$ and information
\begin{equation}
I(a)=\frac{n^2\Var(\Ein)}{a^2(n-a)^2}=\frac{N_{\mathrm{in}}}{a(n-a)}\approx\frac{n}{2qa}.
\label{eq:Ia}
\end{equation}
The estimator $\hat a=n\,\Ein/N_{\mathrm{in}}$ is unbiased with $\Var(\hat a)=a(n-a)/N_{\mathrm{in}}=1/I(a)$, so it attains the Cram\'er--Rao bound. It is the UMVUE, because $\Ein$ is complete and sufficient (binomial family) and $\hat a$ is an unbiased function of it; it is the BLUE, because an unbiased linear estimator $\sum c_{ij}A_{ij}$ over intra-community pairs requires $\sum c_{ij}=n$, has variance $(a/n)(1-a/n)\sum c_{ij}^2$, and $\sum c_{ij}^2\ge(\sum c_{ij})^2/N_{\mathrm{in}}$ with equality only for equal weights; and it is the maximum-likelihood estimator, since the score vanishes at $\hat a$. The same holds for $\hat b=n\,\Eout/N_{\mathrm{out}}$, with $I(b)\approx n(q-1)/(2qb)$.

For the parameters $(d,\lambda)$ we have $\partial_\lambda a=d(q-1)$, $\partial_\lambda b=-d$, $\partial_da=a/d$, $\partial_db=b/d$, and for large $n$ the scores are
\begin{equation}
S_\lambda=d(q-1)\frac{\Ein-\E \Ein}{a}-d\,\frac{\Eout-\E \Eout}{b},\qquad S_d=\frac{\nE-\E \nE}{d}.
\label{eq:scores}
\end{equation}
With $\Var(\Ein)\approx na/(2q)$ and $\Var(\Eout)\approx nb(q-1)/(2q)$,
\begin{equation*}
\Cov(S_d,S_\lambda)=\frac1d\Big[d(q-1)\frac{\Var(\Ein)}a-d\,\frac{\Var(\Eout)}b\Big]=\frac{(q-1)n}{2q}-\frac{(q-1)n}{2q}=0 ,
\end{equation*}
so $d$ and $\lambda$ are orthogonal parameters: ignorance of $d$ costs nothing in the estimation of $\lambda$, and $d$ is estimated by the mean degree. The complete-data information about $\lambda$ is
\begin{equation}
I_c(\lambda)=\Var(S_\lambda)=\frac{nd^2(q-1)}{2q}\Big[\frac{q-1}a+\frac1b\Big]=\frac{nd^3(q-1)}{2ab}=\frac{n\,d\,(q-1)}{2\,(1+(q-1)\lambda)(1-\lambda)},
\label{eq:Ic}
\end{equation}
using $(q-1)b+a=qd$ and $ab=d^2(1+(q-1)\lambda)(1-\lambda)$. For $q=2$ this is $nd/(2(1-\lambda^2))$: each of the $nd/2$ edges is an intra-community edge with probability $(1+\lambda)/2$, and a Bernoulli observation with this parameter carries information $1/(1-\lambda^2)$ about $\lambda$. For the running example, $(1+\frac87)(\frac57)=\frac{75}{49}$ and $I_c=1.829\,n$, so that with known labels $\lambda$ could be estimated with standard error $0.739/\sqrt n$, equal to $0.0074$ for $n=10^4$.

\subsection{The observed graph: score and missing information}
\label{sec:unknown}

The law of the graph is the mixture $P_\lambda(A)=\sum_\sigma\pi_{\mathrm{bal}}(\sigma)P_\sigma(A)$. Differentiating its logarithm,
\begin{equation}
S(A)=\partial_\lambda\ln P_\lambda(A)=\sum_\sigma\frac{\pi_{\mathrm{bal}}(\sigma)P_\sigma(A)}{P_\lambda(A)}\,\partial_\lambda\ln P_\sigma(A)=\E[\,S_\lambda\mid A\,],
\label{eq:fisheridentity}
\end{equation}
the complete-data score averaged over the posterior (Fisher's identity, \citealp{louis1982}). Substituting $\Eout=\nE-\Ein$ in \eqref{eq:scores}, the two centering constants cancel for large $n$ (both equal $d(q-1)n/(2q)$) and
\begin{equation}
S_\lambda=\kappa\,(\Ein-\pbar\,\nE),\qquad \kappa=\frac{d(q-1)}a+\frac db=\frac{qd^2}{ab},\qquad \kappa\pbar=\frac db .
\label{eq:scoreedges}
\end{equation}
Given $A$ the total $\nE$ is known and $\E[\Ein\mid A]=\sum_ep_e$, where $p_e$ is the posterior probability that edge $e$ is an intra-community edge. Hence
\begin{equation}
S(A)=\kappa\sum_e\big(p_e-\pbar\big):
\label{eq:scoregraph}
\end{equation}
the score of the graph is the sum over edges of the excess posterior probability of being an intra-community edge. By the variance decomposition $\Var X=\Var\E[X\mid A]+\E\Var[X\mid A]$ applied to $S_\lambda$,
\begin{equation}
I_{\mathrm{graph}}(\lambda)=I_c(\lambda)-\kappa^2\,\E\big[\Var(\Ein\mid A)\big]=I_c(\lambda)\,(1-V),\qquad
V=\frac{\E[\Var(\Ein\mid A)]}{\Var(\Ein-\pbar\,\nE)}\in[0,1],
\label{eq:missing}
\end{equation}
the second term being the missing information in the sense of \citet{louis1982}. $V=0$ when the graph determines the intra-community edges and $V=1$ when it is uninformative about them. Belief propagation substitutes its own estimate for $p_e$; at the uninformative fixed point $p_e=\pbar$ for every edge (proof of Proposition~\ref{prop:em}), so the algorithm's version of \eqref{eq:scoregraph} vanishes identically: the algorithm behaves as if $V=1$ regardless of the true value of $V$.

\subsection{A Cram\'er--Rao bound within a class of statistics}
\label{sec:rcr}

Let $\cT$ be a linear space of statistics of $A$ containing the constants, and let $S_\cT$ be the $L^2$-projection of the score $S$ onto $\cT$, characterized by $\E[U(S-S_\cT)]=0$ for all $U\in\cT$. Define the information of the class as $I_\cT(\lambda)=\Var(S_\cT)$.

\begin{theorem}[Restricted Cram\'er--Rao bound]
\label{thm:rcr}
For every $T\in\cT$ with $g(\lambda)=\E_\lambda T$,
\begin{equation*}
\Var_\lambda(T)\ \ge\ \frac{g'(\lambda)^2}{I_\cT(\lambda)},\qquad\text{and}\qquad I_\cT(\lambda)\le I_{\mathrm{graph}}(\lambda).
\end{equation*}
Equality in the first inequality holds if and only if $T$ is an affine function of $S_\cT$.
\end{theorem}

\begin{proof}
Differentiating $g(\lambda)=\sum_AT(A)P_\lambda(A)$ gives $g'(\lambda)=\E[TS]=\E[TS_\cT]$, the last step by the projection property since $T\in\cT$. As $1\in\cT$, $\E S_\cT=\E S=0$, so $g'(\lambda)=\Cov(T,S_\cT)$ and Cauchy--Schwarz gives $g'(\lambda)^2\le\Var(T)\,I_\cT(\lambda)$, with equality if and only if $T$ is affine in $S_\cT$. Finally $\E S^2=\E S_\cT^2+\E(S-S_\cT)^2\ge\E S_\cT^2$.
\end{proof}

Theorem~\ref{thm:rcr} contains the Cram\'er--Rao bound ($\cT$ = all statistics), the information of a statistic $Y$ ($\cT$ = functions of $Y$, $S_\cT=\E[S\mid Y]$), the Gauss--Markov bound for the BLUE ($\cT$ = the span of a fixed list of statistics) and a low-degree information ($\cT=\R[A]_{\le D}$). Information is monotone in the class,
\begin{equation}
I_{\mathrm{class}}(\lambda)\ \le\ I_{\mathrm{graph}}(\lambda)\ \le\ I_c(\lambda)=\frac{n\,d\,(q-1)}{2(1+(q-1)\lambda)(1-\lambda)},
\label{eq:ladder}
\end{equation}
and the question whether an efficient method can learn $\lambda$ becomes whether $I_{\mathrm{class}}$ grows with $n$ for a class of efficiently computable statistics.

\subsection{Cycle counts}
\label{sec:cycles}

Let $X_k$ be the number of cycles of length $k\ge3$ in $A$; it is a polynomial of degree $k$ in the entries of $A$, computable by counting closed non-backtracking walks. Throughout Sections~\ref{sec:cycles}--\ref{sec:karlin} the average degree $d$ is treated as known: it is orthogonal to $\lambda$ (Section~\ref{sec:known}) and is estimated by the mean degree. A $k$-cycle is an ordered list of $k$ distinct vertices up to the $2k$ rotations and reflections, and there are $n(n-1)\cdots(n-k+1)\approx n^k$ ordered lists; given the labels, the $k$ edges are present with probability $\prod_ic(\sigma_{v_i},\sigma_{v_{i+1}})/n$, where $c(s,t)=a$ if $s=t$ and $b$ otherwise. Averaging over the labels, $\sum_{s_1,\dots,s_k}q^{-k}\prod_ic(s_i,s_{i+1})=\tr(Q^k)$ with $Q=(c(s,t)/q)_{s,t}$, whose eigenvalues are $d$ (once) and $d\lambda$ ($q-1$ times) by Section~\ref{sec:lambda}. Therefore
\begin{equation}
\mu_k=\E X_k=\frac{d^k}{2k}\big(1+(q-1)\lambda^k\big)+o(1),
\label{eq:muk}
\end{equation}
the Erd\H{o}s--R\'enyi value $d^k/2k$ multiplied by $1+(q-1)\lambda^k$.

\begin{fact}[Poisson limit of cycle counts; \citet{mossel2015,banks2016}]
\label{fact:poisson}
For every fixed $K$, as $n\to\infty$ the vector $(X_3,\dots,X_K)$ converges in law to independent Poisson variables with means \eqref{eq:muk}.
\end{fact}

Fact~\ref{fact:poisson}, which underlies the contiguity argument of \citet{mossel2015} and the small-subgraph conditioning step of \citet{banks2016}, is the only external input of this section. Accordingly, all statements in Sections~\ref{sec:cycles}--\ref{sec:karlin} concern the \emph{limit experiment}, in which $(X_3,\dots,X_K)$ are independent Poisson variables with means \eqref{eq:muk}, for an arbitrary fixed $K$, and all constants are uniform in $K$. For the testing problem of Section~\ref{sec:cycletests} the limit statements transfer to large finite $n$, because convergence in law on a countable set implies convergence in total variation.

For $X\sim\mathrm{Poisson}(\mu(\lambda))$ the score is $\mu'(X-\mu)/\mu$ and the information is $\mu'^2/\mu$. From \eqref{eq:muk}, $\mu_k'=(q-1)d^k\lambda^{k-1}/2$ and
\begin{equation}
I_k=\frac{\mu_k'^2}{\mu_k}=\frac{(q-1)^2\,k\,d^k\lambda^{2k-2}}{2\,(1+(q-1)\lambda^k)}=\frac{(q-1)^2\,k\,r^k}{2\lambda^2\,(1+(q-1)\lambda^k)},
\label{eq:Ik}
\end{equation}
using $d^k\lambda^{2k-2}=r^k/\lambda^2$. Independent counts add their information, and $c_\lambda\le1+(q-1)\lambda^k\le q$ by \eqref{eq:clambda}, so
\begin{equation}
\frac{(q-1)^2}{2q\lambda^2}\sum_{k=3}^Kkr^k\ \le\ I_{\mathrm{cyc},K}(\lambda)=\sum_{k=3}^KI_k\ \le\ \frac{(q-1)^2}{2c_\lambda\lambda^2}\sum_{k=3}^Kkr^k ,
\label{eq:Icycbounds}
\end{equation}
and for $r<1$, from $\sum_{k\ge1}kr^k=r/(1-r)^2$,
\begin{equation}
\sum_{k\ge3}kr^k=\frac{r}{(1-r)^2}-r-2r^2=\frac{r^3(3-2r)}{(1-r)^2}.
\label{eq:sumkrk}
\end{equation}

\begin{theorem}[Information in cycle counts]
\label{thm:cycleinfo}
Let $I_{\mathrm{cyc},K}(\lambda)$ be the Fisher information about $\lambda$ carried by $(X_3,\dots,X_K)$ in the limit experiment. If $r=d\lambda^2<1$, then for every $K$,
\begin{equation*}
I_{\mathrm{cyc},K}(\lambda)\ \le\ \frac{(q-1)^2\,r^3(3-2r)}{2c_\lambda\lambda^2(1-r)^2}=\frac{(q-1)^2\,d\,r^2(3-2r)}{2c_\lambda(1-r)^2},
\end{equation*}
a constant independent of $n$ and $K$, with $c_\lambda=1$ in the assortative case. If $r\ge1$, then $I_{\mathrm{cyc},K}(\lambda)\ge\frac{(q-1)^2}{2q\lambda^2}(3+4+\dots+K)\to\infty$ as $K\to\infty$.
\end{theorem}

\begin{proof}
The upper bound is \eqref{eq:Icycbounds} with \eqref{eq:sumkrk} and $r/\lambda^2=d$; the lower bound follows from $kr^k\ge k$ for $r\ge1$.
\end{proof}

The threshold thus appears as the radius of convergence of the cycle-information series. Before drawing the Cram\'er--Rao consequence we record which functions of $\lambda$ can be estimated without bias from cycle counts at all.

\begin{corollary}[Cram\'er--Rao bound for cycle statistics]
\label{cor:crcyc}
In the limit experiment, let $T$ be a function of $(X_3,\dots,X_K)$ with $\E_\lambda|T|<\infty$ and $\E_\lambda T=g(\lambda)$ throughout the parameter range. Then $g$ is a power series in $\lambda$ whose exponents are sums of elements of $\{3,\dots,K\}$; in particular $\lambda$ and $\lambda^2$ admit no unbiased estimator based on cycle counts, while $\lambda^k$ does for every $k\ge3$. If $d\lambda^2<1$ and $g(\lambda)=\lambda^k$ with $k\ge3$, then
\begin{equation*}
\frac{\Var_\lambda(T)}{\lambda^{2k}}\ \ge\ \frac{k^2}{\lambda^2\,I_{\mathrm{cyc},K}(\lambda)}\ \ge\ \frac{2c_\lambda\,k^2(1-r)^2}{(q-1)^2\,r^3(3-2r)},
\end{equation*}
a bound that depends neither on $n$ nor on $K$: including longer cycles cannot bring the relative error of an unbiased cycle-based estimator of $\lambda^k$ below this constant.
\end{corollary}

\begin{proof}
Multiplying $\E_\lambda T=g(\lambda)$ by $e^{\sum_k\mu_k}$ gives $\sum_xT(x)\prod_k\mu_k^{x_k}/x_k!=g(\lambda)\,e^{\sum_k\mu_k(\lambda)}$. With $\mu_k=\alpha_k+\beta_k\lambda^k$, $\alpha_k=d^k/2k$ and $\beta_k=(q-1)d^k/2k$, the left side is a power series in $\lambda$ containing only exponents that are sums of elements of $\{3,\dots,K\}$, and so is $e^{\sum_k\mu_k(\lambda)}$, whose constant term $e^{\sum_k\alpha_k}$ is nonzero. If $g$ had a term of some other exponent $m$, the smallest such $m$ would appear on the right side with coefficient $g_me^{\sum_k\alpha_k}\ne0$, since every other contribution to that coefficient would express $m$ as a sum of allowed exponents; this contradicts the form of the left side. The exponents $1$ and $2$ are not such sums, while $\lambda^k=(\mu_k-\alpha_k)/\beta_k$ for $k\ge3$. The inequality is Theorem~\ref{thm:rcr} with $g'(\lambda)=k\lambda^{k-1}$ together with Theorem~\ref{thm:cycleinfo} and $\lambda^2d=r$.
\end{proof}

For the running example, $c_\lambda=1$, $(q-1)^2d=22.4$, $r^2=0.013061$, $3-2r=2.7714$ and $(1-r)^2=0.78449$, so $I_{\mathrm{cyc},K}\le0.517$ for every $K$, and the relative standard error of every unbiased cycle-based estimator of $\lambda^3$ is at least $14.6$, whatever cycle lengths it uses; the estimator from triangles alone has relative standard error $16.6$ (Section~\ref{sec:umvue}, Table~\ref{tab:ladder}). In the disassortative family of Section~\ref{sec:expq5} ($q=5$, $a=1$) at $d\lambda^2=0.9$ one has $\lambda=-0.235$ and $c_\lambda=0.948$, so the disassortative bounds of this section differ from their assortative counterparts by about five percent.

\subsection{UMVUE and BLUE from cycle counts}
\label{sec:umvue}

Solving \eqref{eq:muk} for $\lambda^k$ gives the estimator
\begin{equation}
T_k=\frac{2kX_k/d^k-1}{q-1},\qquad\E T_k=\lambda^k,\qquad
\Var(T_k)=\Big(\frac{2k}{(q-1)d^k}\Big)^2\mu_k=\frac{2k\,(1+(q-1)\lambda^k)}{(q-1)^2\,d^k}.
\label{eq:Tk}
\end{equation}
$T_k$ is the unique UMVUE of $\lambda^k$ based on $X_k$, by the Lehmann--Scheff\'e theorem for the Poisson family, and it attains the Cram\'er--Rao bound: the target has slope $2k/((q-1)d^k)$ in $\mu_k$ and the information about $\mu_k$ is $1/\mu_k$, so the bound is $\mathrm{slope}^2\,\mu_k$, which is \eqref{eq:Tk}. Its relative variance is
\begin{equation}
\frac{\Var(T_k)}{\lambda^{2k}}=\frac{2k\,(1+(q-1)\lambda^k)}{(q-1)^2\,r^k}\ \longrightarrow\
\begin{cases}0 & (r>1),\\ \infty & (r<1),\end{cases}\qquad k\to\infty :
\label{eq:relvar}
\end{equation}
long cycles give an increasingly accurate estimate of $\lambda$ above the threshold and an increasingly inaccurate one below it; $T_k^{1/k}$ is the estimator used by \citet{mossel2015} above the threshold for $q=2$. For the running example and $k=3$, $\lambda^3=0.0233$, $\Var(T_3)=6\cdot1.093/(16\cdot2.744)=0.149$, and the standard error $0.387$ is $16.6$ times the quantity estimated.

Several cycle lengths do not form a one-parameter exponential family (their means depend on $\lambda$ through different powers), so there is no one-dimensional complete sufficient statistic and no UMVUE in general. The linear theory is complete, however: for $U=\sum_kw_k(X_k-\mu_k)$, Cauchy--Schwarz gives
\begin{equation*}
\frac{\Cov(U,S)^2}{\Var(U)}=\frac{(\sum_kw_k\mu_k')^2}{\sum_kw_k^2\mu_k}\ \le\ \sum_k\frac{\mu_k'^2}{\mu_k}=I_{\mathrm{cyc},K},
\end{equation*}
with equality for $w_k=\mu_k'/\mu_k$. The BLUE of the cycle counts therefore carries the entire cycle information, and since it is a polynomial of degree $K$, the low-degree information of degree $K$ is at least $I_{\mathrm{cyc},K}$. Above the threshold this lower bound is of order $Kr^K$, so the low-degree information grows without bound with the degree. Below the threshold the cycle counts contribute only a bounded amount to it, and the inequality says nothing about the other polynomials of degree $K$; the question of what they add is a different one, and we do not pursue it here.

\subsection{Tests based on cycle counts}
\label{sec:cycletests}

Consider detection, Definition~\ref{def:overlap}, with cycle counts as the only statistics. For Poisson laws with means $\mu$ and $\nu$, $\sum_x\sqrt{p(x)p_0(x)}=e^{-(\mu+\nu)/2}\sum_x(\sqrt{\mu\nu})^x/x!=e^{-(\sqrt\mu-\sqrt\nu)^2/2}$, and the affinity of independent coordinates is the product of the affinities, so in the limit of Fact~\ref{fact:poisson}
\begin{equation}
\rho_K=\exp\Big(-\frac12\sum_{k=3}^K\big(\sqrt{\mu_k(\lambda)}-\sqrt{\mu_k(0)}\big)^2\Big).
\label{eq:hellinger}
\end{equation}
With $u_k=(q-1)\lambda^k$ we have $\mu_k(\lambda)=\mu_k(0)(1+u_k)$ and $c_\lambda\le1+u_k\le q$ by \eqref{eq:clambda}, so
\begin{equation}
\big|\sqrt{1+u_k}-1\big|=\frac{|u_k|}{\sqrt{1+u_k}+1}\ \le\ \frac{|u_k|}{1+\sqrt{c_\lambda}},
\label{eq:sqrtbound}
\end{equation}
and each term of \eqref{eq:hellinger} is at most $\frac{d^k}{2k}\cdot\frac{(q-1)^2\lambda^{2k}}{(1+\sqrt{c_\lambda})^2}=\frac{(q-1)^2r^k}{2k(1+\sqrt{c_\lambda})^2}$. Since $\sum_{k\ge1}r^k/k=-\ln(1-r)$,
\begin{equation}
\rho_K\ \ge\ \rho_\infty:=\exp\Big(-\frac{(q-1)^2}{4(1+\sqrt{c_\lambda})^2}\big(-\ln(1-r)-r-\tfrac{r^2}2\big)\Big)\ >\ 0\qquad(r<1),
\label{eq:rhobound}
\end{equation}
uniformly in $K$; in the assortative case the prefactor is $(q-1)^2/16$. Writing $|p-p_0|=|\sqrt p-\sqrt{p_0}|(\sqrt p+\sqrt{p_0})$ and applying Cauchy--Schwarz, $\sum|p-p_0|\le\sqrt{(2-2\rho)(2+2\rho)}$, that is $\TV\le\sqrt{1-\rho^2}$.

\begin{theorem}[Cycle-count tests]
\label{thm:cycletest}
Let $\lambda\in(-1/(q-1),1)\setminus\{0\}$. In the limit experiment: if $d\lambda^2<1$, every test of $\mathrm{SBM}(n,q,a,b)$ against $\mathrm{SBM}(n,q,d,d)$ based on $(X_3,\dots,X_K)$ has success probability at most $\frac12\big(1+\sqrt{1-\rho_\infty^2}\big)<1$, for every $K$, and this bound is the limit of the optimal finite-$n$ success probabilities. If $d\lambda^2\ge1$, then $\rho_K\to0$ as $K\to\infty$, so the success probability of the best cycle-count test tends to one in the iterated limit ($n\to\infty$, then $K\to\infty$).
\end{theorem}

\begin{proof}
The first statement combines \eqref{eq:rhobound}, $\TV\le\sqrt{1-\rho^2}$ and the formula $\frac12(1+\TV)$ for the success probability of the optimal test between equiprobable hypotheses. For the second, since $1+u_k\le q$, the identity in \eqref{eq:sqrtbound} gives $|\sqrt{1+u_k}-1|\ge|u_k|/(1+\sqrt q)\ge|u_k|/(2\sqrt q)$ for either sign of $u_k$, so each term of \eqref{eq:hellinger} is at least $(q-1)^2r^k/(8qk)\ge(q-1)^2/(8qk)$ when $r\ge1$; the harmonic series diverges, hence $\rho_K\to0$ and $\TV\ge1-\rho_K\to1$.
\end{proof}

For the running example, $(q-1)^2/16=1$ and $-\ln(1-r)-r-r^2/2=r^3/3+r^4/4+\dots=0.000545$, so $\rho_\infty\ge0.99946$, $\TV\le0.033$, and no test based on cycle counts of bounded length is correct more than $51.6\%$ of the time, against $1/2$ for random guessing. \citet{banks2016} proved that detection is possible below $d\lambda^2=1$ for $q\ge5$ in the disassortative and $q\ge11$ in the assortative case, by a test that searches over partitions. Theorem~\ref{thm:cycletest} applies to both cases, and in particular to the disassortative one in which the window is established: there the cycle statistics are capped strictly below certainty while the full graph admits a consistent test, a second form of the gap, for the class of cycle statistics.

\subsection{Karlin's theorem, range restriction and shrinkage}
\label{sec:karlin}

Karlin's theorem \citep{karlin1958}, see \citet[Section~5.2]{lehmann1998}: in a one-parameter exponential family $\beta(\theta)e^{\theta T(x)}h(x)$ with natural parameter space $(\theta_{\mathrm{low}},\theta_{\mathrm{up}})$, the estimator $\delta_{s,\zeta}=(T+\zeta s)/(1+s)$ of $\gamma(\theta)=\E_\theta T$ is admissible under squared error provided that, for some $\theta_0$,
\begin{equation*}
\int_{\theta_0}^{\theta^*}\frac{e^{-\zeta s\theta}}{\beta(\theta)^s}\dt\theta\to\infty\ \ (\theta^*\to\theta_{\mathrm{up}})
\qquad\text{and}\qquad
\int_{\theta^*}^{\theta_0}\frac{e^{-\zeta s\theta}}{\beta(\theta)^s}\dt\theta\to\infty\ \ (\theta^*\to\theta_{\mathrm{low}}).
\end{equation*}

\begin{proposition}[Karlin's conditions for the Poisson family; range restriction]
\label{prop:karlin}
For the Poisson family the conditions hold for all $s\ge0$ and $\zeta\ge0$; in particular $X$ is admissible for its mean, and so is $T_k$ for $\lambda^k$ whenever $\mu_k$ ranges over $(0,\infty)$. In the block model $\lambda\in[-1/(q-1),1]$, so $\lambda^k$ lies in $J_k=[-(q-1)^{-k},1]$ and $\mu_k$ is confined to $[(1-(q-1)^{1-k})d^k/2k,\,qd^k/2k]$; there $T_k$ takes values outside $J_k$ with positive probability and is dominated by its projection onto $J_k$.
\end{proposition}

\begin{proof}
With $\theta=\ln\mu$, $\Prob(X=x)=\exp(-e^\theta)e^{\theta x}/x!$, so $\beta(\theta)=\exp(-e^\theta)$, $T(x)=x$, $\gamma(\theta)=e^\theta$ and the natural parameter space is $\R$. The integrand is $\exp(-\zeta s\theta+se^\theta)$. As $\theta\to+\infty$, $e^\theta$ eventually dominates $\zeta\theta$, so the integrand is at least $1$ beyond some $\theta_1$ and the integral grows at least linearly; as $\theta\to-\infty$, both $-\zeta s\theta$ and $se^\theta$ are nonnegative for $\theta\le0$, so the integrand is at least $1$ and the integral again grows linearly. Since $T_k$ and $\lambda^k$ are the same increasing affine function of $X_k$ and $\mu_k$, admissibility transfers. On the restricted range the argument of Lemma~\ref{lem:clip} applies.
\end{proof}

For the running example with $k=3$, $\mu_3=\frac{d^3}6(1+4\lambda^3)=0.500$, the target is $\lambda^3=0.0233$, the range is $J_3=[-1/64,1]$, and
\begin{center}
\begin{tabular}{@{}cccc@{}}
\toprule
$X_3$ & probability & $T_3$ & projection onto $J_3$\\
\midrule
$0$ & $0.607$ & $-0.250$ & $-0.016$\\
$1$ & $0.303$ & $0.297$ & $0.297$\\
$2$ & $0.076$ & $0.843$ & $0.843$\\
$\ge3$ & $0.014$ & $\ge1.39$ & $1$\\
\bottomrule
\end{tabular}
\end{center}
The mean squared error is $0.149$ for $T_3$ and $0.0883$ for its projection: the UMVUE reports a value below the smallest one the model allows in $61\%$ of the networks, and clipping lowers its error by $41\%$. Shrinkage leads to the same conclusion: for $cT_k$ with $g=\lambda^k$, $\MSE(c)=c^2\Var(T_k)+(1-c)^2g^2$ is minimized at $c^*=1/(1+\Var(T_k)/g^2)$, and $\Var(T_k)/g^2$ is the relative variance \eqref{eq:relvar}, so $c^*\to1$ as $k$ grows when $r>1$ and $c^*\to0$ when $r<1$. For the example, $\Var(T_3)/g^2=275$ and $c^*=0.0036$: the optimal shrinkage retains $0.36\%$ of the triangle count, and its error is $99.6\%$ of that of ignoring triangles altogether.

\subsection{The information ladder}

\begin{table}[t]
\centering
\caption{Fisher information about $\lambda$ at three levels, with the values for the running example ($q=5$, $a=3$, $b=1$). The cycle-count rows refer to the limit experiment of Fact~\ref{fact:poisson}; the last row refers to the estimated score \eqref{eq:scoregraph} with belief-propagation beliefs, not to a statistic.}
\label{tab:ladder}
\small
\begin{tabular}{@{}p{3.2cm}p{4.8cm}p{2.9cm}p{4.0cm}@{}}
\toprule
Statistics available & Information about $\lambda$ & Growth & Running example\\
\midrule
Complete data (labels known) & $\dfrac{n\,d\,(q-1)}{2(1+(q-1)\lambda)(1-\lambda)}$ & linear in $n$ & $1.829\,n$; s.e.\ $0.0074$ at $n=10^4$\\[2ex]
The observed graph & $I_c\,(1-V)$ & governed by $V$ & ---\\[0.5ex]
Cycle counts, $d\lambda^2<1$ & $\le\dfrac{(q-1)^2\,d\,r^2(3-2r)}{2c_\lambda(1-r)^2}$ & bounded uniformly in $K$ & $\le0.517$; rel.\ s.e.\ of $\lambda^3$ $\ge14.6$\\[2ex]
Cycle counts, $d\lambda^2\ge1$ & $\ge\dfrac{(q-1)^2}{2q\lambda^2}\sum_{k=3}^Kk$ & unbounded in $K$ & ---\\[1.5ex]
Score estimated by belief propagation at its uninformative fixed point & $0$ (identically) & none & $0$\\
\bottomrule
\end{tabular}
\end{table}

Table~\ref{tab:ladder} summarizes the section. The complete-data information grows linearly in $n$; the information of the graph is a fraction $1-V$ of it; the information of cycle counts, in the limit experiment, is bounded uniformly in the cycle length exactly when $d\lambda^2<1$, where it also caps every cycle-count test strictly below certainty, and grows without bound with the cycle length otherwise; and the score estimated by belief propagation at its uninformative fixed point is identically zero. Together with Section~\ref{sec:decision}, this identifies the Kesten--Stigum threshold with the point at which cycle statistics begin to accumulate information as their length grows: below it, cycles of all lengths together carry a bounded amount of information about $\lambda$; above it, each additional cycle length adds more than the one before.

\section{Many communities and PageRank}
\label{sec:many}

\subsection{The recovery condition for $q=n^\chi$}

When the number of communities grows polynomially, $q=n^\chi$ with $0<\chi<1$, the Kesten--Stigum line is no longer the boundary of efficient recovery. \citet[Theorems~1.2, 1.3, 1.7]{chin2025} proved that for $1/2<\chi<1$ a polynomial-time path-counting algorithm achieves weak recovery when $d\lambda^{1/\chi}>C(\log d)^2$, that low-degree polynomials fail below $d\lambda^2=1$ when $q\le n^{1/2-\gamma}$ (Fact~\ref{fact:chin}), and that with unbounded computation recovery is possible when $d\lambda>14$ and impossible when $d\lambda<1$; \citet{carpentier2026,carpentier2026b} established low-degree lower bounds at the new threshold for all densities, with matching algorithms based on counts of paths, cliques and blown-up cycles, up to logarithmic factors. The following computation locates the new threshold in two lines and is the basis of the experiments of Section~\ref{sec:expmany}.

Fix two vertices $u,v$ and count the walks of length $k$ from $u$ to $v$. With $Q$ as in Section~\ref{sec:cycles}, the expected count is $(q/n)(Q^k)_{\sigma_u\sigma_v}$, that is $(d^k+(q-1)(d\lambda)^k)/n$ if $\sigma_u=\sigma_v$ and $(d^k-(d\lambda)^k)/n$ otherwise. The signal is the difference, $q(d\lambda)^k/n$; the noise is the fluctuation of a count with mean of order $d^k/n$, hence of standard deviation $\sqrt{d^k/n}$. The per-pair signal-to-noise ratio is
\begin{equation}
z_k=\frac{q(d\lambda)^k/n}{\sqrt{d^k/n}}=(d\lambda^2)^{k/2}\,\frac{q}{\sqrt n}.
\label{eq:zk}
\end{equation}
The factor $q/\sqrt n$ separates the two regimes. For bounded $q$ it is negligible: a single pair of vertices cannot be tested, and recovery must pool the information of whole neighborhoods, which by \eqref{eq:snr} is possible exactly when $d\lambda^2>1$. For $q\gg\sqrt n$ it is large, and pairs can be tested directly, provided that the walks are long enough to span a community and the test is still reliable at that length: $(d\lambda)^k\approx n/q$ gives $k^*=\ln(n/q)/\ln(d\lambda)$, and $z_{k^*}>1$ becomes, after taking logarithms with $\ln q=\chi\ln n$,
\begin{equation}
\chi\ln d+\ln\lambda>0\quad\Longleftrightarrow\quad d^\chi\lambda>1\quad\Longleftrightarrow\quad d\lambda^{1/\chi}>1 ,
\label{eq:chinnumber}
\end{equation}
the condition of \citet{chin2025} up to their logarithmic factor. Two efficient mechanisms thus coexist, collective (spectral, belief propagation; $d\lambda^2>1$) and pairwise (path counting; $d^\chi\lambda>1$), and the weaker condition governs. The resulting threshold for efficient recovery, established on the achievability side by \citet{chin2025} up to logarithmic factors and conjectured to be sharp for low-degree algorithms \citep{chin2025,carpentier2026}, is
\begin{equation}
\lambda>d^{-\max(1/2,\,\chi)},\qquad\chi=\frac{\ln q}{\ln n},
\label{eq:maxformula}
\end{equation}
the two conditions coinciding at $q=\sqrt n$. The number of communities enters only through the logarithmic ratio $\chi$. For every $\chi<1$ a gap remains between the information-theoretic condition $d\lambda>1$ and \eqref{eq:maxformula}, and it closes as $\chi\to1$.

\subsection{PageRank as a walk count with suboptimal weights}
\label{sec:pagerank}

Personalized PageRank from a vertex $u$ assigns to each vertex $v$ the score
\begin{equation*}
\pi_u(v)=(1-\alpha)\sum_{k\ge0}\alpha^k(P^k)_{uv},\qquad P=\Diag(\deg)^{-1}A,
\end{equation*}
a geometrically weighted sum over walk lengths of the probability that the random walk started at $u$ is at $v$ after $k$ steps. On a sparse graph with degrees concentrated near $d$, $(P^k)_{uv}\approx d^{-k}(A^k)_{uv}$, so $\pi_u$ is a walk count with weight $(\alpha/d)^k$ on length $k$. The matched filter for \eqref{eq:zk} weights each length by its signal $(d\lambda)^k$ divided by its noise variance $d^k$, that is by $\lambda^k$. The two coincide only if $\alpha=d\lambda$, which is impossible with $\alpha<1$ whenever $d\lambda>1$: PageRank discounts long walks too strongly in exactly the regime in which they carry the signal. This extends to the sparse, unseeded, many-community setting the result of \citet{kloumann2017}, who showed for the dense two-block model with a seed set that the optimal geometric discount is $(p_{\mathrm{in}}-p_{\mathrm{out}})/(p_{\mathrm{in}}+p_{\mathrm{out}})$, which is $\lambda$ for $q=2$. Three consequences follow and are confirmed in Section~\ref{sec:expmany}: ordinary PageRank, a single score per vertex, cannot separate pairs and tracks the degree; personalized PageRank with every vertex as its own seed is a valid pair test for many communities but is dominated by the $\lambda^k$-weighted path count, increasingly so as $d$ grows; and neither method enters the fixed-$q$ hard window, since after replacing degrees by their mean both are low-degree rules to which Theorem~\ref{thm:lowdeg} and Fact~\ref{fact:chin} apply.

\section{Local learning and the curse of dimensionality}
\label{sec:local}

The methods that fail in the hard window, nearest-neighbor voting, spectral embeddings followed by $k$-means, random-walk and Markov-chain methods, and message-passing neural networks, share one feature: each of them reads the graph through bounded neighborhoods. This section places that observation in the framework of learning theory. The starting point is a lower bound for learning labels that carry no structure at all, which fixes the baseline that any structured problem must beat.

\begin{proposition}[Lower bound for unstructured labels]
\label{prop:nfl}
Let the input space consist of $p^D$ points ($D$ features with $p$ values each), let each point carry a label drawn independently and uniformly from $K$ classes, and let a learner observe the labels of $n$ distinct points and then predict the label of a point drawn uniformly from the input space. Every learner has expected error at least
\begin{equation}
\Big(1-\frac1K\Big)\Big(1-\frac{n}{p^D}\Big).
\label{eq:unstructured}
\end{equation}
In particular, for $K=2$ an error of $\epsilon$ requires $n\ge(1-2\epsilon)\,2^D$ labeled points.
\end{proposition}

\begin{proof}
With probability at least $1-n/p^D$ the test point is not among the $n$ observed points; its label is then independent of the observations and uniform on $K$ classes, so any prediction is wrong with probability $1-1/K$.
\end{proof}

Proposition~\ref{prop:nfl} is the curse of dimensionality in its sharpest form: each additional feature multiplies the number of labels required by $p$. It does not, however, locate the Kesten--Stigum threshold, for three reasons that are instructive. It is a statement about information, not computation, so it cannot separate what is possible from what is possible efficiently. It assumes the labels of distinct points to be unrelated, whereas in the block model the graph \emph{is} the relation between labels: every edge is a noisy statement that its endpoints agree. And it contains none of $a$, $b$, $q$ or $d$. Its role is to describe the baseline: a learner without usable structure guesses, and guessing is what the trivial risk $\vq$ measures in the pair loss.

What the block model adds is a specific structure, and \eqref{eq:snr} says precisely which learners can exploit it. Table~\ref{tab:local} lists the familiar methods, the quantity each one reads, and where it stops.

\begin{table}[t]
\centering
\caption{What local learning methods read in the graph, and why they stop at the Kesten--Stigum line.}
\label{tab:local}
\small
\begin{tabular}{@{}p{3.6cm}p{5.2cm}p{6.6cm}@{}}
\toprule
Method & Quantity read & Limitation\\
\midrule
Nearest-neighbor voting (label propagation) & majority of the labels in a depth-$k$ neighborhood & by \eqref{eq:snr} the signal of the vote decays geometrically with depth when $d\lambda^2<1$; revealed labels add a gain proportional to their fraction (Section~\ref{sec:expq2}; \citealp{zhang2014semi,kanade2016})\\
$k$-means on a spectral embedding & a few eigenvectors of a matrix built from $A$ & only as good as the embedding: adjacency eigenvectors fail even above the threshold in sparse graphs, Bethe-Hessian eigenvectors work down to it \citep{krzakala2013,saade2014}\\
Random walks, PageRank, Markov-chain embeddings & polynomials of $P=\Diag(\deg)^{-1}A$ & low-degree rules (Section~\ref{sec:gap}); PageRank in addition uses suboptimal weights (Section~\ref{sec:pagerank})\\
Gibbs sampling, simulated annealing & local moves on the posterior landscape & trapped in the uninformative basin below, and well above, the threshold (Section~\ref{sec:expq5}; \citealp{angelini2023})\\
Message-passing (graph) neural networks & $k$ rounds of neighborhood aggregation & learned aggregation cannot create information that the neighborhood does not contain; $k$ rounds are a degree-$k$ rule when the aggregation is polynomial\\
\bottomrule
\end{tabular}
\end{table}

The escape from the curse of dimensionality is more data per unknown. The block model offers the same escape: at fixed $\lambda$, increasing the degree $d$ raises $d\lambda^2$ past $1$, after which the Bethe-Hessian and belief-propagation methods succeed. The hard window is the regime in which $d$, $\lambda$ and $q$ are fixed by the data, the information is present in the global consistency of a labeling, and no bounded neighborhood contains it.

\section{Experiments}
\label{sec:experiments}

We test every quantitative statement of Sections~\ref{sec:threshold}--\ref{sec:many} on simulated networks. The algorithms are specified in Appendix~\ref{app:algorithms}. Unless stated otherwise, $n=50{,}000$, labels are i.i.d.\ uniform, one network is generated per parameter setting, and the overlap \eqref{eq:overlap} is computed exactly by solving the assignment problem over the $q!$ renamings. Away from the threshold the overlap of a method varies from one network to another by about $n^{-1/2}\approx0.005$ at $n=50{,}000$, so that differences of a few hundredths between methods are significant there. Near $d\lambda^2=1$ the variation between networks of the same size is larger, several hundredths (Table~\ref{tab:q2sizes}), and the departures from chance of that size that appear just below the threshold in Tables~\ref{tab:q2} and~\ref{tab:q5} are real features of networks of this size; the relevant question is whether they shrink as $n$ grows, and Section~\ref{sec:expq2} shows that they do.

\subsection{Two communities: the threshold at $d\lambda^2=1$}
\label{sec:expq2}

For $q=2$ and $d=3$ every method is at or within a few hundredths of chance ($0.50$) below $d\lambda^2=1$ and departs from chance above it (Table~\ref{tab:q2}; Figure~\ref{fig:scans}, left), with belief propagation from a random initialization performing as well as from the true labels, in agreement with \citet{mossel2015,mossel2018,massoulie2014}. At a single network size the transition is smoothed: the Bethe-Hessian overlap of $0.541$ at $d\lambda^2=0.9$ in Table~\ref{tab:q2} is well outside the $n^{-1/2}$ fluctuations. Table~\ref{tab:q2sizes} therefore follows the two best methods across three network sizes, with several networks per setting. Below the threshold the overlaps move toward $0.50$ as $n$ grows: at $d\lambda^2=0.9$ the Bethe-Hessian overlap falls from $0.533$ at $n=12{,}500$ to $0.522$ at $50{,}000$ and $0.514$ at $200{,}000$, and at $0.95$ from $0.567$ to $0.549$ and $0.519$, with belief propagation following the same course ($0.527$, $0.530$, $0.502$ and $0.543$, $0.545$, $0.511$). Above the threshold they do not: the overlaps at $1.1$ and $1.2$ are $0.61$--$0.62$ and $0.66$--$0.69$ at all three sizes. At $d\lambda^2=1$ itself, where the limiting overlap is $0.50$ but the approach to it is slowest, the three sizes are indistinguishable. The transition thus steepens with the size of the network and is located at $d\lambda^2=1$, as Theorem~\ref{thm:bp} and the results of \citet{mossel2015,mossel2018} require. The eigenvectors of the adjacency matrix remain at chance even at $d\lambda^2=2$, being localized on high-degree vertices \citep{krzakala2013}, while the Bethe-Hessian eigenvectors track belief propagation \citep{saade2014}. Label propagation with a fraction of revealed labels gains, below the threshold, only in proportion to that fraction (overlap $0.52$ with $1\%$ and $0.57$ with $5\%$ revealed at $d\lambda^2=0.5$), consistent with \citet{zhang2014semi} and \citet{kanade2016}.

\begin{table}[t]
\centering
\caption{Two communities, $d=3$, $n=50{,}000$: overlap of four methods (chance $0.50$).}
\label{tab:q2}
\small
\begin{tabular}{@{}ccccc@{}}
\toprule
$d\lambda^2$ & Bethe Hessian + $k$-means & adjacency + $k$-means & BP, random init. & BP, planted init.\\
\midrule
0.5 & 0.506 & 0.501 & 0.504 & 0.504\\
0.9 & 0.541 & 0.504 & 0.503 & 0.519\\
1.1 & 0.638 & 0.502 & 0.644 & 0.644\\
1.5 & 0.773 & 0.503 & 0.800 & 0.800\\
2.0 & 0.876 & 0.503 & 0.895 & 0.895\\
\bottomrule
\end{tabular}
\end{table}

\begin{table}[t]
\centering
\caption{Two communities, $d=3$: overlap of the Bethe-Hessian method and of belief propagation from a random initialization at three network sizes, averaged over $4$, $2$ and $2$ networks respectively (chance $0.50$). The range over networks of the same size shrinks with $n$: up to $0.09$ at $n=12{,}500$ (and $0.15$ for belief propagation at $1.1$, where one network had not left the uninformative fixed point after $1{,}500$ iterations), up to $0.08$ at $50{,}000$, and below $0.015$ at $200{,}000$ except at $0.95$ and $1.0$, where it reaches $0.07$.}
\label{tab:q2sizes}
\small
\begin{tabular}{@{}ccccccc@{}}
\toprule
& \multicolumn{3}{c}{Bethe Hessian + $k$-means} & \multicolumn{3}{c}{BP, random init.}\\
\cmidrule(lr){2-4}\cmidrule(lr){5-7}
$d\lambda^2$ & $n=12{,}500$ & $n=50{,}000$ & $n=200{,}000$ & $n=12{,}500$ & $n=50{,}000$ & $n=200{,}000$\\
\midrule
0.70 & 0.536 & 0.505 & 0.504 & 0.520 & 0.502 & 0.507\\
0.80 & 0.505 & 0.510 & 0.504 & 0.507 & 0.508 & 0.505\\
0.90 & 0.533 & 0.522 & 0.514 & 0.527 & 0.530 & 0.502\\
0.95 & 0.567 & 0.549 & 0.519 & 0.543 & 0.545 & 0.511\\
1.00 & 0.569 & 0.550 & 0.552 & 0.545 & 0.533 & 0.542\\
1.05 & 0.604 & 0.596 & 0.573 & 0.598 & 0.570 & 0.577\\
1.10 & 0.624 & 0.611 & 0.614 & 0.605 & 0.627 & 0.623\\
1.20 & 0.680 & 0.660 & 0.668 & 0.690 & 0.679 & 0.685\\
\bottomrule
\end{tabular}
\end{table}

\begin{figure}[t]
\centering
\includegraphics[width=\textwidth]{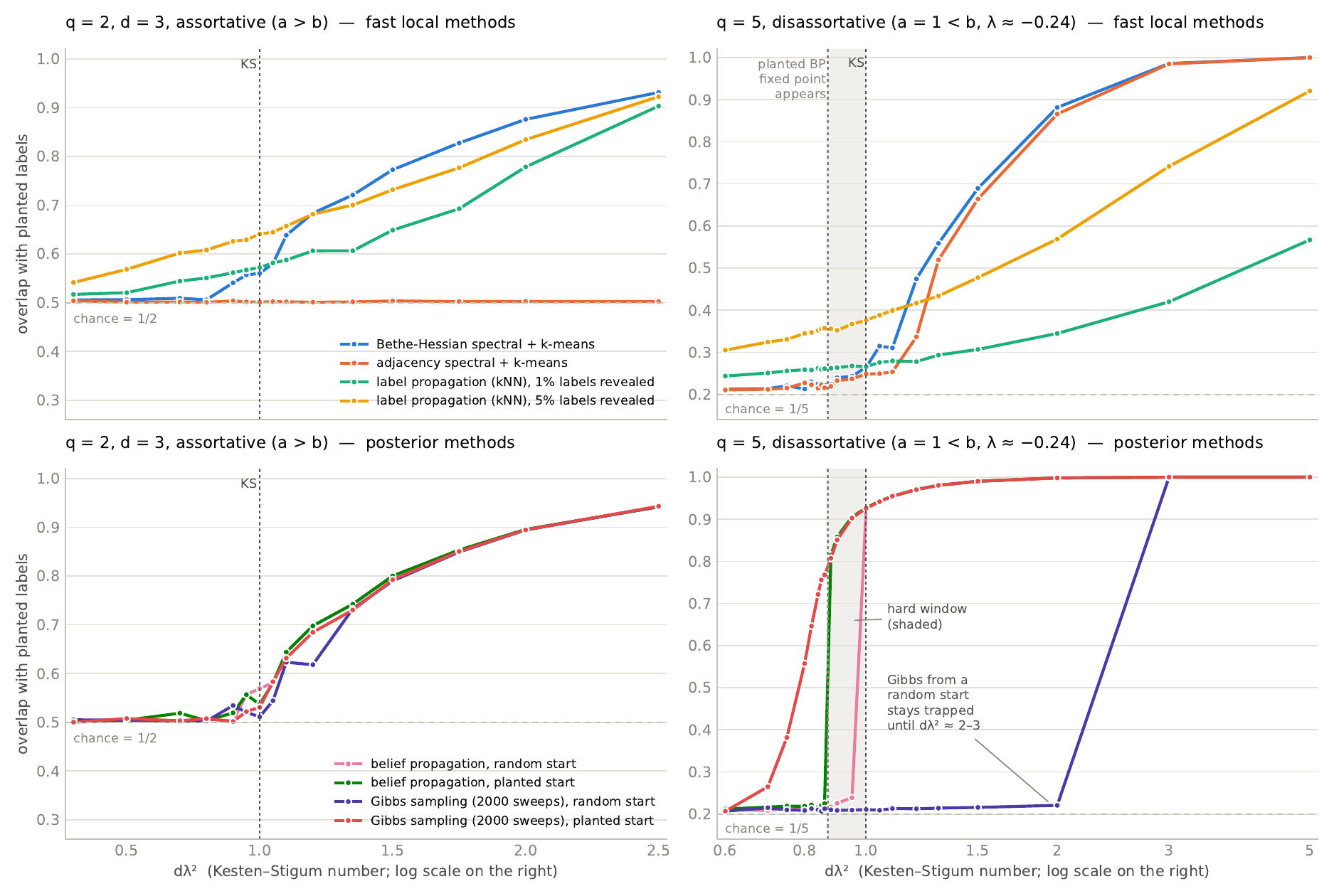}
\caption{Overlap against $d\lambda^2$ for $n=50{,}000$. Left: $q=2$, $d=3$, assortative. Right: $q=5$, disassortative ($a=1$; $b$ chosen to give the value of $d\lambda^2$ on the axis, so that $d$ grows from $11.5$ to $82$). Top: efficient local methods; bottom: belief propagation and Gibbs sampling from a random initialization and from the true labels. The shaded band is the window in which the true labeling is a locally stable fixed point of belief propagation while every method initialized without information remains at chance; it contains the information-theoretically hard window.}
\label{fig:scans}
\end{figure}

\subsection{Five communities: the hard window}
\label{sec:expq5}

For $q=5$ we use the disassortative family $a=1<b$, for which \citet{banks2016} prove detection below the threshold. Table~\ref{tab:q5} and Figure~\ref{fig:scans} (right) exhibit the two basins of Section~\ref{sec:bp}: at $d\lambda^2=0.88$ and $0.95$ the true labeling is a locally stable fixed point of belief propagation (planted initialization: overlap $0.82$--$0.91$), while belief propagation from a random initialization, the spectral method and Gibbs sampling from a random initialization all remain within a few hundredths of chance ($0.20$); belief propagation from a random initialization succeeds at $d\lambda^2=1$, where Theorem~\ref{thm:bp} predicts the loss of stability of the uninformative fixed point. Figure~\ref{fig:traj} shows the dynamics. At $d\lambda^2=0.6$ the Gibbs sampler initialized at the truth decays to chance within $10^3$ sweeps, so the truth is not even a metastable state. At $0.9$ the same sampler remains near overlap $0.78$ for $3\times10^4$ sweeps when initialized at the truth and at $0.21$ when initialized at random: the posterior has two basins on the same graph. At $1.5$ belief propagation recovers the communities in about $50$ iterations while the Gibbs sampler is still at chance after $3\times10^4$ sweeps, and it remains trapped up to $d\lambda^2\approx2$--$3$, in agreement with the analysis of simulated annealing in planted problems by \citet{angelini2023}. In the assortative families with $q=5$ and $q=10$ the window is too narrow to be resolved at this size, consistent with its being certified by \citet{banks2016} only for $q\ge11$ and by \citet{mossel2023} at large degree.

\begin{table}[t]
\centering
\caption{Five communities, disassortative ($a=1$), $n=50{,}000$: overlap (chance $0.20$). The Gibbs columns are after $2{,}000$ sweeps (one sweep updates every vertex once): at $0.80$ the planted initialization is still decaying toward chance, the truth not being a fixed point of belief propagation there, while at $0.88$ and $0.95$ the planted values are stable and the random initializations are trapped (Figure~\ref{fig:traj}).}
\label{tab:q5}
\footnotesize
\begin{tabular}{@{}cccccc@{}}
\toprule
$d\lambda^2$ & Bethe Hessian + $k$-means & BP, random & BP, planted & Gibbs, random & Gibbs, planted\\
\midrule
0.60 & 0.212 & 0.211 & 0.212 & 0.207 & 0.207\\
0.80 & 0.213 & 0.214 & 0.218 & 0.208 & 0.557\\
0.88 & 0.224 & 0.217 & 0.815 & 0.210 & 0.807\\
0.95 & 0.243 & 0.238 & 0.906 & 0.209 & 0.903\\
1.00 & 0.264 & 0.927 & 0.927 & 0.211 & 0.926\\
1.50 & 0.689 & 0.990 & 0.990 & 0.216 & 0.990\\
2.00 & 0.881 & 0.998 & 0.998 & 0.221 & 0.998\\
3.00 & 0.986 & 1.000 & 1.000 & 1.000 & 1.000\\
\bottomrule
\end{tabular}
\end{table}

\begin{figure}[t]
\centering
\includegraphics[width=\textwidth]{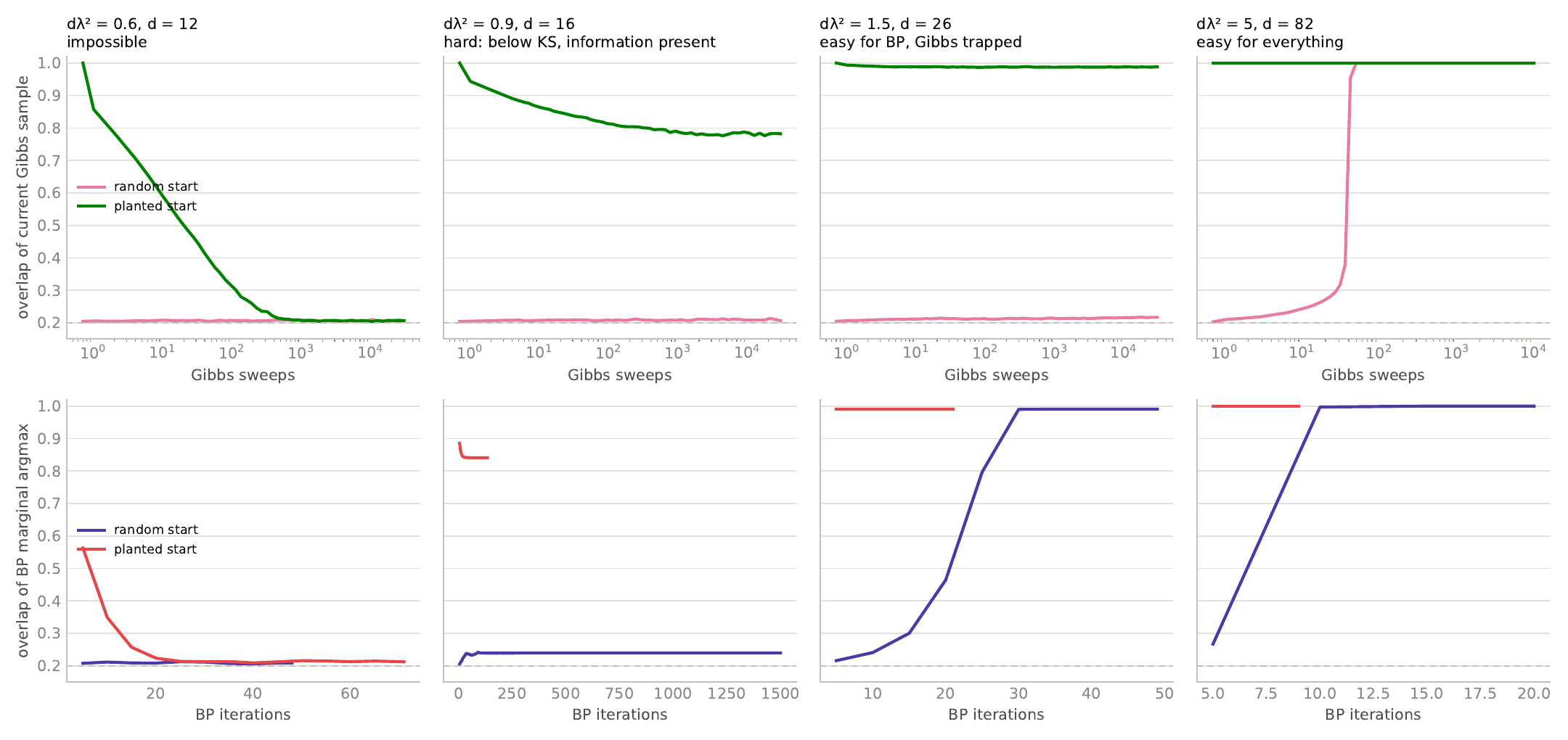}
\caption{Five communities, disassortative, $n=50{,}000$: overlap of the current Gibbs sample (top; sweeps on a logarithmic scale) and of the belief-propagation marginals (bottom) against time, from a random initialization and from the true labels, at four values of $d\lambda^2$.}
\label{fig:traj}
\end{figure}

\paragraph{Planted coloring.}
In the planted $5$-coloring model ($a=0$, $\lambda=-1/4$) the true labeling became a stable fixed point of belief propagation between $d\lambda^2=0.80$ and $0.81$, that is for an average degree between $12.80$ and $12.96$. The value computed by \citet{decelle2011} from the cavity method is $c_d\approx12.84$.

\paragraph{The running example.}
For $q=5$, $a=3$, $b=1$ ($d\lambda^2=0.114$) all methods ended at overlaps between $0.20$ and $0.22$, within $0.02$ of chance: spectral $0.205$, label propagation with $5\%$ of labels revealed $0.222$, belief propagation $0.204$, Gibbs sampling $0.204$, and Gibbs sampling initialized at the truth decayed to chance. The fraction of isolated vertices was $24.6\%$, against $e^{-1.4}=24.7\%$.

\subsection{Exact and quantitative checks of Sections~\ref{sec:decision} and~\ref{sec:bp}}
\label{sec:expexact}

\begin{figure}[t]
\centering
\includegraphics[width=\textwidth]{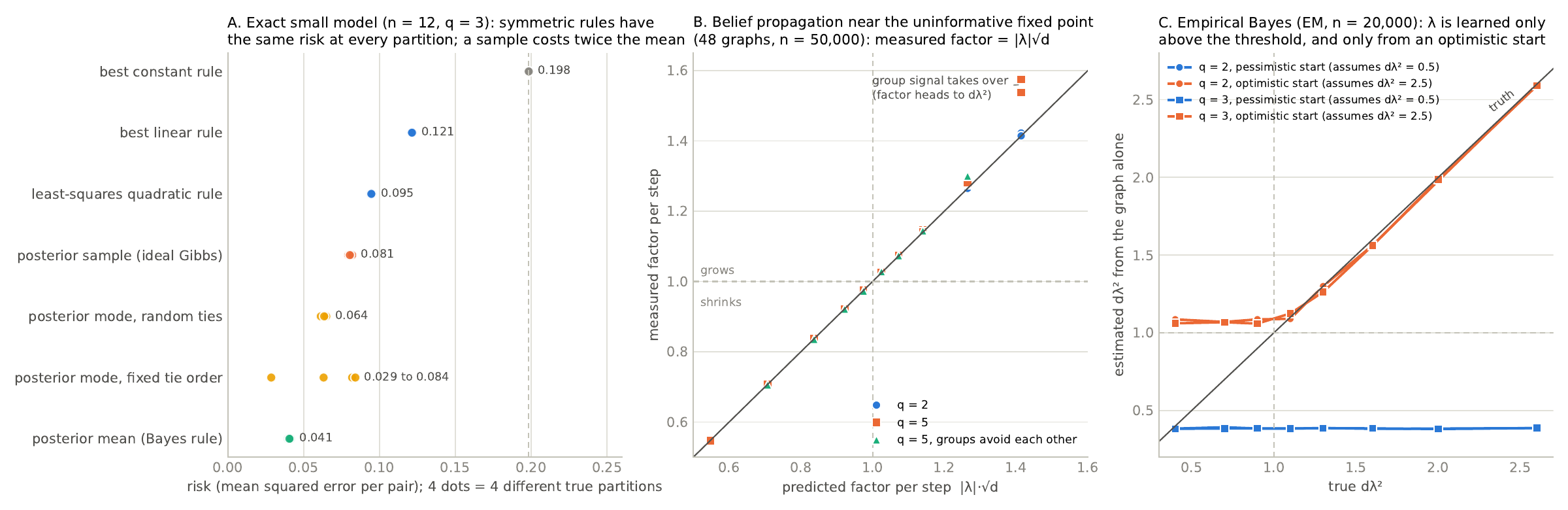}
\caption{A: risk of seven rules at four true partitions in the exact model ($n=12$, $q=3$, $d=4$, $\lambda=0.6$; $10^4$ graphs per partition). B: measured growth factor of a perturbation of belief propagation at the uninformative fixed point against the prediction $|\lambda|\sqrt d$ of Theorem~\ref{thm:bp}, on 48 graphs with $n=50{,}000$. C: $d\lambda^2$ estimated by EM from the graph alone against its true value, from a pessimistic and an optimistic initialization ($n=20{,}000$).}
\label{fig:checks}
\end{figure}

\paragraph{An exactly solvable model.}
For $n=12$ and $q=3$ the balanced orbit contains $5{,}775$ partitions, and the posterior mean, the posterior mode and posterior samples under $\pi_{\mathrm{bal}}$ are computed exactly; every rule is evaluated at four true partitions on $10^4$ graphs each (Figure~\ref{fig:checks}A, Table~\ref{tab:exact}). In this model the target is centered at its orbit mean, $x_{ij}=\ind{\sigma_i=\sigma_j}-c$ with $c=(n/q-1)/(n-1)=3/11$, so that the best constant rule is $t\equiv0$, with risk $c(1-c)=0.198$; all other rules considered contain a constant term, so their risks are unaffected by the centering. Every equivariant rule has the same risk at the four partitions up to sampling error (Lemma~\ref{lem:sym}), whereas the rule that always reports a fixed partition has risk $0$ there and $0.36$--$0.45$ elsewhere, $0.397$ on average; its random relabeling has risk $0.397$ at every partition (Theorem~\ref{thm:worst}(b)) and its symmetrization is the best constant rule, with risk $0.198$ (Lemma~\ref{lem:sym}). The ratio of the risk of a posterior sample to that of the posterior mean is $2.000$, $1.984$ and $1.973$ at $\lambda=0.3$, $0.6$ and $0.8$ (Lemma~\ref{lem:sample}). The risks of the nested classes are ordered as constant $\ge$ linear $\ge$ quadratic $\ge$ all rules. The maximum-likelihood partition (Proposition~\ref{prop:map}) differs from the true partition in $28.1\%$ of the graphs at $\lambda=0.6$, against $34.8\%$ for a posterior sample, but has pair risk $0.064$ against $0.041$ for the posterior mean, and at $\lambda=0.3$ its pair risk $0.295$ exceeds the risk $0.198$ of the best constant rule. Clipping the fitted polynomial rules (Lemma~\ref{lem:clip}) lowers their risk by up to $10\%$ (degree~$2$ at $\lambda=0.8$: $37\%$ of outputs outside the range; risk $0.0377\to0.0338$).

\begin{table}[t]
\centering
\caption{Exact model ($n=12$, $q=3$, $d=4$, $\lambda=0.6$): risk of seven rules at four true partitions.}
\label{tab:exact}
\small
\begin{tabular}{@{}p{5.6cm}p{3.2cm}p{6.6cm}@{}}
\toprule
Rule & Risk & Comment\\
\midrule
Best constant rule & 0.198 & $c(1-c)$\\
Best linear rule & 0.121 & $\cP_1$\\
Least-squares quadratic rule (ten features) & 0.094--0.095 & larger class, smaller risk\\
Posterior sample (ideal Gibbs) & 0.079--0.082 & $2\times$ Bayes risk (Lemma~\ref{lem:sample})\\
Posterior mode, random tie-breaking & 0.061--0.065 & equivariant, constant risk; Bayes for the $0$--$1$ loss only\\
Posterior mode, fixed tie-breaking & 0.029--0.084 & not equivariant: risk depends on the partition\\
Posterior mean & 0.040--0.041 & smallest average and worst-case risk; constant\\
\bottomrule
\end{tabular}
\end{table}

\paragraph{Growth factor of belief propagation.}
On 48 graphs with $n=50{,}000$ ($q\in\{2,5\}$, $d\in[3,28]$, both signs of $\lambda$, $d\lambda^2\in[0.3,2]$), belief propagation was initialized at $1/q$ plus noise of magnitude $10^{-8}$ and the norm of the perturbation was tracked for 16 iterations. On the 39 graphs with $d\lambda^2\le1.3$ the measured factor per iteration agrees with $|\lambda|\sqrt d$ to within $0.006$ (mean absolute deviation $0.002$), and on all 48 graphs it exceeds $1$ exactly when $d\lambda^2>1$ (Figure~\ref{fig:checks}B); for $d\lambda^2\in\{1.6,2\}$ the measured factor exceeds $|\lambda|\sqrt d$ by up to $0.16$, the community signal, which grows by the larger factor $d\lambda^2$, taking over within the 16 iterations. Corollary~\ref{cor:rate} is confirmed by the 13 random-initialization runs of Sections~\ref{sec:expq2}--\ref{sec:expq5} with $d\lambda^2\le0.88$ ($q=2$: $0.3$--$0.8$; $q=5$: $0.6$--$0.88$; initial noise $0.02$, tolerance $10^{-5}$, $\alpha=0.8$): the number of iterations ranges from 18 to 118 while the product $\text{iterations}\times\ln(1/f)$ stays between $6.0$ and $8.1$, against $\ln(0.02/10^{-5})=7.6$.

\paragraph{Empirical Bayes.}
On 16 graphs with $n=20{,}000$ ($q=2$, $d=3$ and $q=3$, $d=5$; true $d\lambda^2\in[0.4,2.6]$), $a$ and $b$ were hidden and estimated by \eqref{eq:em}--\eqref{eq:em2} from a pessimistic initialization (assumed $d\lambda^2=0.5$) and an optimistic one (assumed $2.5$), both with the degree misspecified by a factor $1.3$ (Figure~\ref{fig:checks}C). In all 16 pessimistic runs $\lambda$ changed by less than $10^{-5}$ and the degree estimate equaled the observed mean degree after one step (Proposition~\ref{prop:em}). From the optimistic initialization, $\lambda$ was recovered to within $0.009$ in all eight runs with true $d\lambda^2\ge1.3$, whereas for the six runs with true values $0.4$, $0.7$ and $0.9$ the estimate converged to $d\lambda^2\in[1.06,1.09]$ irrespective of the truth, and for the two runs at $1.1$ it converged to $1.09$ and $1.12$, indistinguishable from that limit: below the threshold $\lambda$ is not learned from either initialization. In the pessimistic runs with true $d\lambda^2\ge1.3$ the residual bias of the beliefs, although of magnitude $10^{-5}$, pointed toward the true communities: the label of largest belief agreed with the truth for $60$--$94\%$ of the vertices for $q=2$ (chance $50\%$) and for $50$--$73\%$ for $q=3$ (chance $33\%$), as expected from repeated application of the operator \eqref{eq:bplin}.

\subsection{Many communities and PageRank}
\label{sec:expmany}

We generated 341 networks from the balanced model with $n\in\{3\times10^4,10^5,3\times10^5\}$, $q=\operatorname{round}(n^\chi)$ for $\chi\in\{0.3,0.5,0.6,0.7,0.8,0.9\}$, $d\in\{4,9,16\}$ ($d\in\{4,16\}$ for $n=3\times10^5$) and 6--8 values of $\lambda$ per family. For 50 source vertices per network, every other vertex was scored by the weighted non-backtracking walk count $\sum_{k\le K}\lambda^kN_k(u,v)$, with $N_k$ computed by the recursion $N_{k+1}=AN_k-(\Diag(\deg)-I)N_{k-1}$ and $K$ between 7 and 13, and we recorded the fraction of true community members among the $n/q-1$ highest scores (the \emph{retrieval fraction}; chance is $(n/q-1)/(n-1)$). Figure~\ref{fig:many} summarizes the results.

\begin{figure}[t]
\centering
\includegraphics[width=\textwidth]{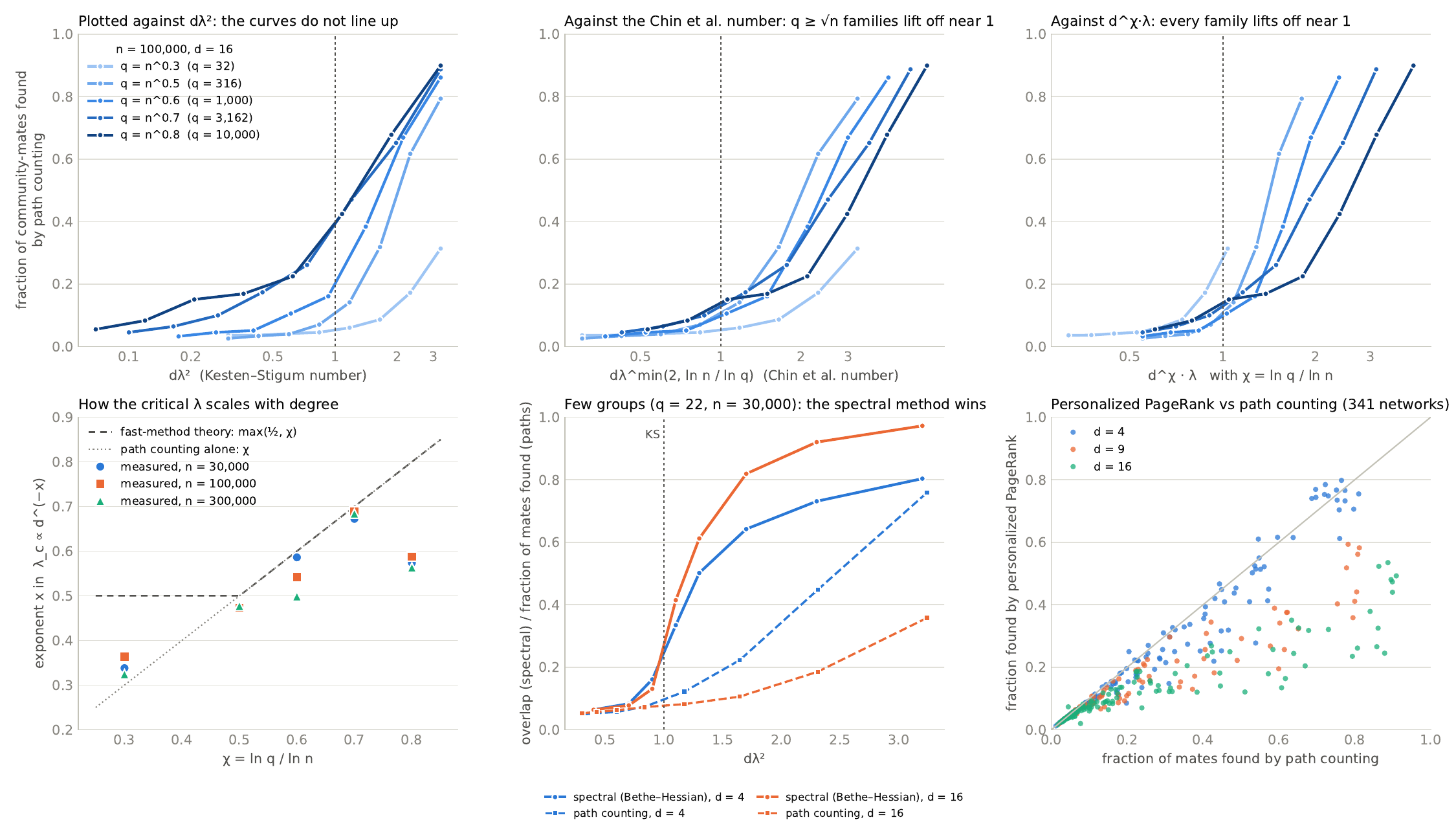}
\caption{Many communities, $q=n^\chi$. Top: retrieval fraction of $\lambda^k$-weighted path counting ($n=10^5$, $d=16$) against $d\lambda^2$ (the families do not align), against $d\lambda^{\min(2,1/\chi)}$ (the families with $q\ge\sqrt n$ depart from chance near $1$) and against $d^\chi\lambda$ (all families depart near $1$). Bottom left: fitted exponent $x$ in $\lambda_c\propto d^{-x}$ against $\chi$, with the predictions $x=\chi$ for path counting and $\max(1/2,\chi)$ for the best efficient method. Bottom middle: for $q=22$, $n=3\times10^4$, the spectral method departs from chance at $d\lambda^2=1$ while path counting requires $d\lambda^2\approx1.5$--$2.5$. Bottom right: personalized PageRank against path counting on all 341 networks.}
\label{fig:many}
\end{figure}

\begin{enumerate}[leftmargin=2em]
\item \emph{Pairwise recovery below the Kesten--Stigum line requires many communities.} At $n=3\times10^5$, $d=4$ and $d\lambda^2\in[0.74,0.78]$ the retrieval fraction was $0.035$ for $\chi=0.3$ ($44$ communities; chance $0.023$), $0.046$ for $\chi=0.5$ ($548$; chance $0.002$), $0.090$ for $\chi=0.6$, $0.128$ for $\chi=0.7$ and $0.226$ for $\chi=0.8$ ($24{,}082$ communities of $12$ vertices; chance $4\times10^{-5}$), against $0.023$, $0.005$, $0.011$, $0.035$ and $0.127$ for direct neighbors alone.
\item \emph{The critical $\lambda$ scales as \eqref{eq:chinnumber} predicts.} Defining $\lambda_c$ by a retrieval fraction of $20\%$ and fitting $\lambda_c\propto d^{-x}$ over the available degrees, the exponent $x$ was $0.32$--$0.36$ for $\chi=0.3$, $0.47$--$0.48$ for $0.5$, $0.50$--$0.59$ for $0.6$, $0.67$--$0.69$ for $0.7$ and $0.56$--$0.59$ for $0.8$ (ranges over $n$), against the prediction $x=\chi$ for path counting; the agreement holds up to $\chi=0.7$ and breaks down at $0.8$, where communities have about ten members and the asymptotic reasoning no longer applies ($\chi=0.9$, with communities of 3--4 vertices, is excluded from the fit for the same reason). For $\chi\le0.7$ the product $d^\chi\lambda_c$ lies between $0.88$ and $1.36$ across all sizes and degrees.
\item \emph{The number $d\lambda^{1/\chi}$ orders the outcomes best.} The rank correlation between the retrieval fraction and a candidate number over the 341 networks is $0.95$ for $d\lambda^{\min(2,1/\chi)}$, $0.93$ for $d^\chi\lambda$, $0.77$ for $d\lambda^2$ and $0.63$ for $d\lambda$.
\item \emph{For few communities the collective method dominates.} For $q=22$ and $n=3\times10^4$ the Bethe-Hessian method departs from chance at $d\lambda^2=1$ for $d=4$ and $d=16$ (overlap $0.13$--$0.16$ at $0.9$, $0.33$--$0.42$ at $1.1$), whereas path counting on the same graphs requires $d\lambda^2\approx1.5$--$2.5$.
\item \emph{PageRank.} Ordinary PageRank is correlated $0.99$ with the degree, and the probability that it ranks a random community member above a random non-member (the area under the ROC curve) is $0.50$, that of random guessing. Personalized PageRank with every vertex as its own seed (restart probabilities $0.5$ and $0.85$, scores normalized by degree) is a valid pair test but is dominated by $\lambda^k$-weighted path counting, increasingly so with $d$: on the networks where path counting retrieves $15$--$70\%$ of the members, the mean retrieval fraction was $0.32$ against $0.36$ at $d=4$, $0.21$ against $0.36$ at $d=9$ and $0.17$ against $0.35$ at $d=16$, as Section~\ref{sec:pagerank} predicts; for few communities below the threshold it remains at chance ($0.033$ against $0.023$ for $44$ communities at $n=3\times10^5$).
\end{enumerate}

\section{Discussion}
\label{sec:discussion}

The results of this paper describe the Kesten--Stigum threshold and the information--computation gap of the sparse block model in the language of classical statistics. Three descriptions emerge, and they agree. In decision-theoretic terms, the uniform prior on each community profile is least favorable for every symmetric class of rules; the posterior mean is the unique Bayes rule, is admissible, and is minimax on the balanced profile; and the Bayes risk of the low-degree class is $\vq(1-\Corr_D^2)$, so that the low-degree bound of \citet{chin2025} and the information-theoretic results of \citet{banks2016,abbe2018,mossel2023} combine into a worst-case theorem: in the hard window no low-degree rule has worst-case risk below the trivial risk, while an exponential-time rule beats it by the explicit margin $\varepsilon^4/q^2$. In information-theoretic terms, the threshold is the radius of convergence of the cycle-information series, the point at which the Fisher information carried by cycle counts, bounded uniformly in the cycle length below the threshold, begins to grow without bound as longer cycles are included, with explicit Cram\'er--Rao constants and an explicit cap on every cycle-count test. In algorithmic terms, it is the point at which the uninformative fixed point of belief propagation loses stability, with growth factor $|\lambda|\sqrt d$, and the EM algorithm, if started below it, confirms its own initialization. The many-community formula \eqref{eq:maxformula} and the analysis of PageRank follow from the same signal-to-noise computation, and the experiments confirm each prediction quantitatively, including the planted-coloring constant of \citet{decelle2011} and the scaling exponents of \citet{chin2025}.

The decision-theoretic framework also explains, in a precise sense, why the hardness of the window is a statement about classes of statistics. Theorem~\ref{thm:worst}(b) reduces the worst-case question to the Bayes risk for \emph{every} class of algorithms, polynomial-time algorithms included; what distinguishes the low-degree class and the cycle statistics is that both their Bayes risk and their Fisher information are $L^2$ projections and can be computed. Extending the restricted Cram\'er--Rao bound of Theorem~\ref{thm:rcr} to classes that capture more of polynomial-time computation, and determining whether the missing-information fraction $V$ of \eqref{eq:missing} is bounded away from one throughout the hard window (which would follow from a consistent estimator of $\lambda$ there), are natural next steps. Two further directions suggested by the experiments are the behavior of clipped low-degree rules, which lie outside $\cP_D$ and dominate the polynomial rules they are built from, and the residual alignment of belief propagation with the truth near its uninformative fixed point above the threshold, which is of order $10^{-5}$ in magnitude but correctly signed.

\bibliographystyle{plainnat}
\bibliography{refs}

@article{holland1983,
  author  = {Holland, Paul W. and Laskey, Kathryn Blackmond and Leinhardt, Samuel},
  title   = {Stochastic blockmodels: First steps},
  journal = {Social Networks},
  volume  = {5},
  number  = {2},
  pages   = {109--137},
  year    = {1983}
}

@article{kesten1966,
  author  = {Kesten, Harry and Stigum, Bernt P.},
  title   = {Additional limit theorems for indecomposable multidimensional {G}alton--{W}atson processes},
  journal = {Annals of Mathematical Statistics},
  volume  = {37},
  number  = {6},
  pages   = {1463--1481},
  year    = {1966}
}

@article{evans2000,
  author  = {Evans, William and Kenyon, Claire and Peres, Yuval and Schulman, Leonard J.},
  title   = {Broadcasting on trees and the {I}sing model},
  journal = {Annals of Applied Probability},
  volume  = {10},
  number  = {2},
  pages   = {410--433},
  year    = {2000}
}

@article{decelle2011,
  author  = {Decelle, Aurelien and Krzakala, Florent and Moore, Cristopher and Zdeborov{\'a}, Lenka},
  title   = {Asymptotic analysis of the stochastic block model for modular networks and its algorithmic applications},
  journal = {Physical Review E},
  volume  = {84},
  pages   = {066106},
  year    = {2011}
}

@article{mossel2015,
  author  = {Mossel, Elchanan and Neeman, Joe and Sly, Allan},
  title   = {Reconstruction and estimation in the planted partition model},
  journal = {Probability Theory and Related Fields},
  volume  = {162},
  pages   = {431--461},
  year    = {2015}
}

@article{mossel2018,
  author  = {Mossel, Elchanan and Neeman, Joe and Sly, Allan},
  title   = {A proof of the block model threshold conjecture},
  journal = {Combinatorica},
  volume  = {38},
  number  = {3},
  pages   = {665--708},
  year    = {2018}
}

@inproceedings{massoulie2014,
  author    = {Massouli{\'e}, Laurent},
  title     = {Community detection thresholds and the weak {R}amanujan property},
  booktitle = {Proceedings of the 46th Annual ACM Symposium on Theory of Computing (STOC)},
  pages     = {694--703},
  year      = {2014}
}

@article{bordenave2018,
  author  = {Bordenave, Charles and Lelarge, Marc and Massouli{\'e}, Laurent},
  title   = {Nonbacktracking spectrum of random graphs: Community detection and nonregular {R}amanujan graphs},
  journal = {Annals of Probability},
  volume  = {46},
  number  = {1},
  pages   = {1--71},
  year    = {2018}
}

@article{abbe2018,
  author  = {Abbe, Emmanuel and Sandon, Colin},
  title   = {Proof of the achievability conjectures for the general stochastic block model},
  journal = {Communications on Pure and Applied Mathematics},
  volume  = {71},
  number  = {7},
  pages   = {1334--1406},
  year    = {2018}
}

@inproceedings{banks2016,
  author    = {Banks, Jess and Moore, Cristopher and Neeman, Joe and Netrapalli, Praneeth},
  title     = {Information-theoretic thresholds for community detection in sparse networks},
  booktitle = {Proceedings of the 29th Conference on Learning Theory (COLT), PMLR 49},
  pages     = {383--416},
  year      = {2016}
}

@article{sly2011,
  author  = {Sly, Allan},
  title   = {Reconstruction for the {P}otts model},
  journal = {Annals of Probability},
  volume  = {39},
  number  = {4},
  pages   = {1365--1406},
  year    = {2011}
}

@inproceedings{mossel2023,
  author    = {Mossel, Elchanan and Sly, Allan and Sohn, Youngtak},
  title     = {Exact phase transitions for stochastic block models and reconstruction on trees},
  booktitle = {Proceedings of the 55th Annual ACM Symposium on Theory of Computing (STOC)},
  pages     = {96--102},
  year      = {2023}
}

@inproceedings{hopkins2017,
  author    = {Hopkins, Samuel B. and Steurer, David},
  title     = {Efficient {B}ayesian estimation from few samples: Community detection and related problems},
  booktitle = {58th Annual IEEE Symposium on Foundations of Computer Science (FOCS)},
  pages     = {379--390},
  year      = {2017}
}

@article{schramm2022,
  author  = {Schramm, Tselil and Wein, Alexander S.},
  title   = {Computational barriers to estimation from low-degree polynomials},
  journal = {Annals of Statistics},
  volume  = {50},
  number  = {3},
  pages   = {1833--1858},
  year    = {2022}
}

@misc{kunisky2019,
  author = {Kunisky, Dmitriy and Wein, Alexander S. and Bandeira, Afonso S.},
  title  = {Notes on computational hardness of hypothesis testing: Predictions using the low-degree likelihood ratio},
  year   = {2019},
  note   = {arXiv:1907.11636}
}

@inproceedings{chin2025,
  author    = {Chin, Byron and Mossel, Elchanan and Sohn, Youngtak and Wein, Alexander S.},
  title     = {Stochastic block models with many communities and the {K}esten--{S}tigum bound},
  booktitle = {Proceedings of the 38th Conference on Learning Theory (COLT), PMLR 291},
  year      = {2025},
  note      = {arXiv:2503.03047}
}

@inproceedings{carpentier2026,
  author    = {Carpentier, Alexandra and Giraud, Christophe and Verzelen, Nicolas},
  title     = {Phase transition for stochastic block model with more than $\sqrt{n}$ communities},
  booktitle = {Proceedings of the 39th Conference on Learning Theory (COLT), PMLR 336},
  year      = {2026},
  note      = {arXiv:2509.15822}
}

@misc{carpentier2026b,
  author = {Carpentier, Alexandra and Giraud, Christophe and Verzelen, Nicolas},
  title  = {Phase transition for stochastic block model with more than $\sqrt{n}$ communities ({II})},
  year   = {2025},
  note   = {arXiv:2511.21526}
}

@article{kloumann2017,
  author  = {Kloumann, Isabel M. and Ugander, Johan and Kleinberg, Jon},
  title   = {Block models and personalized {P}age{R}ank},
  journal = {Proceedings of the National Academy of Sciences},
  volume  = {114},
  number  = {1},
  pages   = {33--38},
  year    = {2017}
}

@article{kawamoto2018,
  author  = {Kawamoto, Tatsuro},
  title   = {Algorithmic detectability threshold of the stochastic block model},
  journal = {Physical Review E},
  volume  = {97},
  pages   = {032301},
  year    = {2018}
}

@article{angelini2023,
  author  = {Angelini, Maria Chiara and Ricci-Tersenghi, Federico},
  title   = {Limits and performances of algorithms based on simulated annealing in solving sparse hard inference problems},
  journal = {Physical Review X},
  volume  = {13},
  pages   = {021011},
  year    = {2023}
}

@inproceedings{saade2014,
  author    = {Saade, Alaa and Krzakala, Florent and Zdeborov{\'a}, Lenka},
  title     = {Spectral clustering of graphs with the {B}ethe {H}essian},
  booktitle = {Advances in Neural Information Processing Systems 27 (NeurIPS)},
  pages     = {406--414},
  year      = {2014}
}

@article{krzakala2013,
  author  = {Krzakala, Florent and Moore, Cristopher and Mossel, Elchanan and Neeman, Joe and Sly, Allan and Zdeborov{\'a}, Lenka and Zhang, Pan},
  title   = {Spectral redemption in clustering sparse networks},
  journal = {Proceedings of the National Academy of Sciences},
  volume  = {110},
  number  = {52},
  pages   = {20935--20940},
  year    = {2013}
}

@article{zhang2014semi,
  author  = {Zhang, Pan and Moore, Cristopher and Zdeborov{\'a}, Lenka},
  title   = {Phase transitions in semisupervised clustering of sparse networks},
  journal = {Physical Review E},
  volume  = {90},
  pages   = {052802},
  year    = {2014}
}

@article{zhang2016,
  author  = {Zhang, Anderson Y. and Zhou, Harrison H.},
  title   = {Minimax rates of community detection in stochastic block models},
  journal = {Annals of Statistics},
  volume  = {44},
  number  = {5},
  pages   = {2252--2280},
  year    = {2016}
}

@article{abbe2017,
  author  = {Abbe, Emmanuel},
  title   = {Community detection and stochastic block models: Recent developments},
  journal = {Journal of Machine Learning Research},
  volume  = {18},
  number  = {177},
  pages   = {1--86},
  year    = {2017}
}

@article{zdeborova2016,
  author  = {Zdeborov{\'a}, Lenka and Krzakala, Florent},
  title   = {Statistical physics of inference: Thresholds and algorithms},
  journal = {Advances in Physics},
  volume  = {65},
  number  = {5},
  pages   = {453--552},
  year    = {2016}
}

@book{lehmann1998,
  author    = {Lehmann, Erich L. and Casella, George},
  title     = {Theory of Point Estimation},
  edition   = {2nd},
  publisher = {Springer},
  address   = {New York},
  year      = {1998}
}

@article{karlin1958,
  author  = {Karlin, Samuel},
  title   = {Admissibility for estimation with quadratic loss},
  journal = {Annals of Mathematical Statistics},
  volume  = {29},
  number  = {2},
  pages   = {406--436},
  year    = {1958}
}

@article{lehmann1950,
  author  = {Lehmann, Erich L. and Scheff{\'e}, Henry},
  title   = {Completeness, similar regions, and unbiased estimation---{P}art {I}},
  journal = {Sankhy\={a}},
  volume  = {10},
  pages   = {305--340},
  year    = {1950}
}

@article{blyth1951,
  author  = {Blyth, Colin R.},
  title   = {On minimax statistical decision procedures and their admissibility},
  journal = {Annals of Mathematical Statistics},
  volume  = {22},
  number  = {1},
  pages   = {22--42},
  year    = {1951}
}

@article{kanade2016,
  author  = {Kanade, Varun and Mossel, Elchanan and Schramm, Tselil},
  title   = {Global and local information in clustering labeled block models},
  journal = {IEEE Transactions on Information Theory},
  volume  = {62},
  number  = {10},
  pages   = {5906--5917},
  year    = {2016}
}

@article{dempster1977,
  author  = {Dempster, Arthur P. and Laird, Nan M. and Rubin, Donald B.},
  title   = {Maximum likelihood from incomplete data via the {EM} algorithm},
  journal = {Journal of the Royal Statistical Society, Series B},
  volume  = {39},
  number  = {1},
  pages   = {1--38},
  year    = {1977}
}

@article{louis1982,
  author  = {Louis, Thomas A.},
  title   = {Finding the observed information matrix when using the {EM} algorithm},
  journal = {Journal of the Royal Statistical Society, Series B},
  volume  = {44},
  number  = {2},
  pages   = {226--233},
  year    = {1982}
}

@book{bollobas2001,
  author    = {Bollob{\'a}s, B{\'e}la},
  title     = {Random Graphs},
  edition   = {2nd},
  publisher = {Cambridge University Press},
  year      = {2001}
}

\appendix
\section{Notation}
\label{app:notation}

\begin{center}
\small
\begin{tabular}{@{}ll@{}}
\toprule
Symbol & Meaning\\
\midrule
$n$, $q$ & number of vertices; number of communities\\
$a/n$, $b/n$ & intra- and inter-community edge probabilities\\
$d=(a+(q-1)b)/q$ & average degree\\
$\lambda=(a-b)/(a+(q-1)b)$ & signal strength; second eigenvalue of the label chain $M$ \eqref{eq:M}\\
$r=d\lambda^2$ & Kesten--Stigum number\\
$c_\lambda=1-(q-1)\max(0,-\lambda)^3$ & sign constant of Section~\ref{sec:fisher} \eqref{eq:clambda}; equals $1$ for $\lambda>0$\\
$\chi=\ln q/\ln n$ & growth exponent for $q=n^\chi$\\
$\sigma\in\Lab$, $A$ & labeling; adjacency matrix\\
$m(\sigma)$, $\Theta_m$, $\Theta_{\mathrm{bal}}$ & community profile; its orbit; the balanced orbit\\
$\pi$, $\pi_m$ & i.i.d.\ uniform prior; uniform prior on $\Theta_m$\\
$x_{ij}=\ind{\sigma_i=\sigma_j}-1/q$ & pair target \eqref{eq:target}\\
$L$, $\risk$, $\brisk$ & pair loss \eqref{eq:loss}; risk; Bayes risk\\
$\vq=\frac1q(1-\frac1q)$, $\triv(\sigma)$ & trivial risk under $\pi$; trivial risk at $\sigma$\\
$\Lam$ & labelings with typical profile \eqref{eq:typical}\\
$\delta^*$, $\delta^*_{\mathrm{bal}}$ & posterior means under $\pi$ and $\pi_{\mathrm{bal}}$\\
$\cP_D$, $\Corr_D$ & low-degree rules of degree $D$; low-degree correlation \eqref{eq:corr}\\
$\psi(i\to j,t)$, $h_t$ & belief-propagation message and field \eqref{eq:bp}\\
$\pbar=a/(qd)$ & probability that an edge is intra-community\\
$N_{\mathrm{in}},N_{\mathrm{out}}$; $\Ein,\Eout,\nE$ & numbers of intra-/inter-community pairs; numbers of edges\\
$S$, $S_\lambda$, $S_d$ & score of the graph; complete-data scores for $\lambda$ and $d$\\
$I_c$, $I_{\mathrm{graph}}$, $I_\cT$, $I_{\mathrm{cyc},K}$ & Fisher information: complete data; graph; class $\cT$; cycles of length $\le K$\\
$V$ & missing-information fraction, $I_{\mathrm{graph}}=I_c(1-V)$\\
$\kappa=qd^2/(ab)$ & constant in $S=\kappa\sum_e(p_e-\pbar)$\\
$X_k$, $\mu_k$, $I_k$ & number of $k$-cycles; its mean \eqref{eq:muk}; its information \eqref{eq:Ik}\\
$T_k$ & UMVUE of $\lambda^k$ from $X_k$ \eqref{eq:Tk}\\
$\rho_K$, $\rho_\infty$, $\TV$ & Hellinger affinities; total variation distance\\
$s$, $\zeta$ & constants of Karlin's estimator $(T+\zeta s)/(1+s)$\\
$z_k$ & per-pair signal-to-noise ratio of walk counts \eqref{eq:zk}\\
\bottomrule
\end{tabular}
\end{center}

\section{Algorithms}
\label{app:algorithms}

\paragraph{Graph generation.}
Labels are i.i.d.\ uniform in Sections~\ref{sec:expq2}--\ref{sec:expq5} and in the belief-propagation and EM checks of Section~\ref{sec:expexact}; the exact model of Section~\ref{sec:expexact} and the networks of Section~\ref{sec:expmany} use the balanced model (community sizes differing by at most one). For each pair of communities the number of edges is binomial with the appropriate number of pairs and probability, and the edges are placed uniformly; multiple edges and loops are discarded.

\paragraph{Belief propagation.}
Equation \eqref{eq:bp} in logarithmic form, with $h_t$ recomputed at every iteration, damping $\alpha=0.8$ (new message $=\alpha\times$update $+(1-\alpha)\times$old; $\alpha=1$ in the growth-factor test), at most $1{,}500$ iterations in the scans and $400$ in the EM runs, and termination when the largest change of a message falls below $10^{-5}$ in the scans and $10^{-6}$ otherwise. Random initialization: $1/q$ plus Gaussian noise of standard deviation $0.02$, clipped and renormalized; planted initialization: $0.9$ on the true label and $0.1/q$ on the others. Output: the label of largest marginal at each vertex. The case $a=0$ is handled by $\log\max(\cdot,10^{-300})$.

\paragraph{Gibbs sampling.}
Random-scan Gibbs sampling of the posterior, which is proportional to $\prod_{\{i,j\}\in E}c_{\sigma_i\sigma_j}\prod_{\{i,j\}\notin E}(1-c_{\sigma_i\sigma_j}/n)$ with $E$ the edge set, in its Poisson form,
\begin{equation*}
\log\Prob(\sigma_i=t\mid\text{rest})=\ln(a/b)\cdot\#\{\text{neighbors of $i$ with label $t$}\}-(a-b)\,m_t/n+\text{const},
\end{equation*}
$m_t$ being the current size of community $t$. A sweep is $n$ single-vertex updates. Output: the label of largest time-averaged marginal over the second half of the run.

\paragraph{Spectral clustering.}
Bethe Hessian $H(\varrho)=(\varrho^2-1)I-\varrho A+\Diag(\deg)$ with $\varrho=\sgn(\lambda)\sqrt d$, the $q$ eigenvectors of smallest eigenvalue, and $k$-means with $10$ restarts. The adjacency variant uses the $q$ eigenvectors of largest (assortative) or smallest (disassortative) eigenvalue of $A$.

\paragraph{Label propagation.}
A fraction $f_{\mathrm{rev}}$ of the vertices have their labels revealed and clamped. The $n\times q$ matrix $F$ of soft labels is initialized at $1/q$ and updated $60$ times by $F\leftarrow\Diag(\deg)^{-1}AFW$, with $W_{st}=c_{st}/\sum_sc_{st}$ the conditional law of a neighbor's label given one's own (for $a<b$ a neighbor's label is evidence against one's own). Overlap is measured on the unrevealed vertices; vertices receiving no signal are labeled at random.

\paragraph{Path counting.}
Non-backtracking walk counts from a source $u$ by $N_1=Ae_u$, $N_2=AN_1-\Diag(\deg)e_u$, $N_{k+1}=AN_k-(\Diag(\deg)-I)N_{k-1}$, combined as $\sum_{k\le K}\lambda^kN_k$ with the true $\lambda$ and $K=\min(14,\max(4,\lceil\ln n/\ln d\rceil+3))$. The retrieval fraction is the proportion of true community members among the $n/q-1$ highest-scoring vertices, averaged over $50$ sources.

\paragraph{PageRank.}
Ordinary PageRank with restart probability $0.15$ and $60$ power iterations. Personalized PageRank from $u$: $\sum_{t\le25}(1-\alpha)\alpha^t(P^{\mathsf T})^te_u$ for $\alpha\in\{0.5,0.85\}$, divided by the degree of the scored vertex; the better $\alpha$ is reported.

\paragraph{EM.}
Equations \eqref{eq:em}--\eqref{eq:em2} with the two-point beliefs from the converged messages, for at most $40$ steps or until $\lambda$ changes by less than $2\times10^{-4}$.

\paragraph{Exact model.}
$n=12$, $q=3$: all $5{,}775$ balanced partitions are enumerated and the posterior mean, mode and samples are computed exactly for $10^4$ graphs per true partition; the target is centered at the orbit mean $c=(n/q-1)/(n-1)$. The linear rule uses the four degree-$1$ pair features invariant under renumbering of the other vertices ($1$, $A_{ij}$, the number of edges at $i$ or $j$, the number of edges at neither); the quadratic rule adds six degree-$2$ features. Coefficients are fitted by least squares on $6{,}000$ training graphs drawn from the prior and evaluated on independent graphs.

\paragraph{Computation.}
All runs use a single CPU core: about one hour for the 341 networks of Section~\ref{sec:expmany}, about $20$ minutes for each scan of Figure~\ref{fig:scans}, about $30$ minutes for the networks of Table~\ref{tab:q2sizes}, and $25$ minutes for the spectral runs with $q=22$; the longest runs are the $3\times10^4$-sweep Gibbs trajectories of Figure~\ref{fig:traj}.

\end{document}